\documentclass[11pt,letterpaper]{amsart}

\usepackage[foot]{amsaddr}
\usepackage{amsmath}
\usepackage{amsthm}
\usepackage{amssymb}
\usepackage{url}
\usepackage{hyperref}
\usepackage{todonotes}
\usepackage{mathtools}
\usepackage{mathrsfs}
\usepackage{booktabs}
\usepackage[top=1in,bottom=1in,left=1in,right=1in]{geometry}
\pdfoutput=1

\newtheorem{Thm}{Theorem}
\newtheorem{Lem}[Thm]{Lemma}

\theoremstyle{remark}

\theoremstyle{definition}

\newenvironment{Proof}{\begin{proof}}{\end{proof}}

\newcommand{\E}{\ensuremath{\mathbb{E}}}
\newcommand{\F}{\ensuremath{\mathbb{F}}}
\newcommand{\Z}{\ensuremath{\mathbb{Z}}}
\newcommand{\bra}{\ensuremath{\langle}}
\newcommand{\ket}{\ensuremath{\rangle}}
\newcommand{\sgn}{\ensuremath{\operatorname{sgn}}}
\newcommand{\per}{\ensuremath{\operatorname{per}}}
\newcommand{\pcc}{\ensuremath{\operatorname{pcc}}}

\begin{document}


\setcounter{page}{0}

\title[]{The Shortest Even Cycle Problem is Tractable}

\author{Andreas Bj\"orklund}
\address{Lund, Sweden}
\email{andreas.bjorklund@yahoo.se}
\author{Thore Husfeldt}
\address{Lund University and Basic Algorithms Research Copenhagen, ITU Copenhagen}
\email{thore@itu.dk}
\author{Petteri Kaski}
\address{Department of Computer Science, Aalto University, Espoo, Finland}
\email{petteri.kaski@aalto.fi}

\begin{abstract}
Given a directed graph as input, we show how to efficiently find a shortest (directed, simple) cycle on an even number of vertices. As far as we know, no polynomial-time algorithm was previously known for this problem. 
In fact, finding \emph{any} even cycle in a directed graph in polynomial time was open for more than two decades until Robertson, Seymour, and Thomas (\emph{Ann.~of~Math.~(2)}\/~1999) and, independently, McCuaig (\emph{Electron.~J.~Combin.}~2004; announced jointly at STOC~1997) gave an efficiently testable structural characterisation of even-cycle-free directed graphs.

Methodologically, our algorithm relies on the standard framework of algebraic fingerprinting and randomized polynomial identity testing over a finite field, and in fact relies on a generating polynomial implicit in a paper of Vazirani and Yannakakis ({\em Discrete~Appl.~Math.}~1989) that enumerates weighted cycle covers by the parity of their number of cycles as a difference of a permanent and a determinant polynomial. The need to work with the permanent---known to be \#P-hard apart from a very restricted choice of coefficient rings (Valiant, {\em Theoret.~Comput.~Sci.} 1979)---is where our main technical contribution occurs. We design a family of finite commutative rings of characteristic~4 that simultaneously (i) give a nondegenerate representation for the generating polynomial identity via the permanent and the determinant, (ii) support efficient permanent computations by extension of Valiant's techniques, and (iii) enable {\em emulation} of finite-field arithmetic in characteristic 2. Here our work is foreshadowed by that of Bj\"orklund and Husfeldt ({\em SIAM~J.~Comput.}~2019), who used a considerably less efficient commutative ring design---in particular, one lacking finite-field emulation---to obtain a polynomial-time algorithm for the shortest two disjoint paths problem in undirected graphs. 

Building on work of Gilbert and Tarjan ({\em Numer.~Math.}~1978) as well as Alon and Yuster ({\em J.~ACM}~2013), we also show how ideas from the nested dissection technique for solving linear equation systems---introduced by George ({\em SIAM J.~Numer.~Anal.}~1973) for symmetric positive definite real matrices---leads to faster algorithm designs in our present finite-ring randomized context when we have control on the separator structure of the input graph; for example, this happens when the input has bounded genus.
\end{abstract}

\maketitle
\thispagestyle{empty}


\clearpage

\section{Introduction}

Given a directed graph, we show how to efficiently find a shortest 
(directed, simple) cycle on an even number of vertices. That this problem has 
eluded tractability until now is perhaps, for lack of a better word, odd.

After all, elementary considerations show that the Shortest \emph{Odd} Cycle problem%
\footnote{%
  To fix terminology, a cycle is a closed walk without repeated vertices; our graphs are unweighted, and they are directed unless otherwise noted; cycles in directed graphs must follow the direction of their edges (sometimes called directed cycles or dicycles), and an odd walk is a walk with an odd number of vertices.}
is tractable. Indeed, every shortest closed odd walk in a directed graph is 
simple, because otherwise the walk would decompose into two shorter closed 
walks that cannot both be even. Such a shortest closed odd walk is a shortest 
odd cycle, and thus it can be found in polynomial time.%
\footnote{More concretely, a breadth-first search from every vertex in a graph with $n$ vertices and $m$ edges finds a shortest odd cycle in time $O(nm)$.}{}
This approach however fails in the even case, because a shortest 
closed \emph{even} walk need not be simple:
\[
\includegraphics[width=20mm]{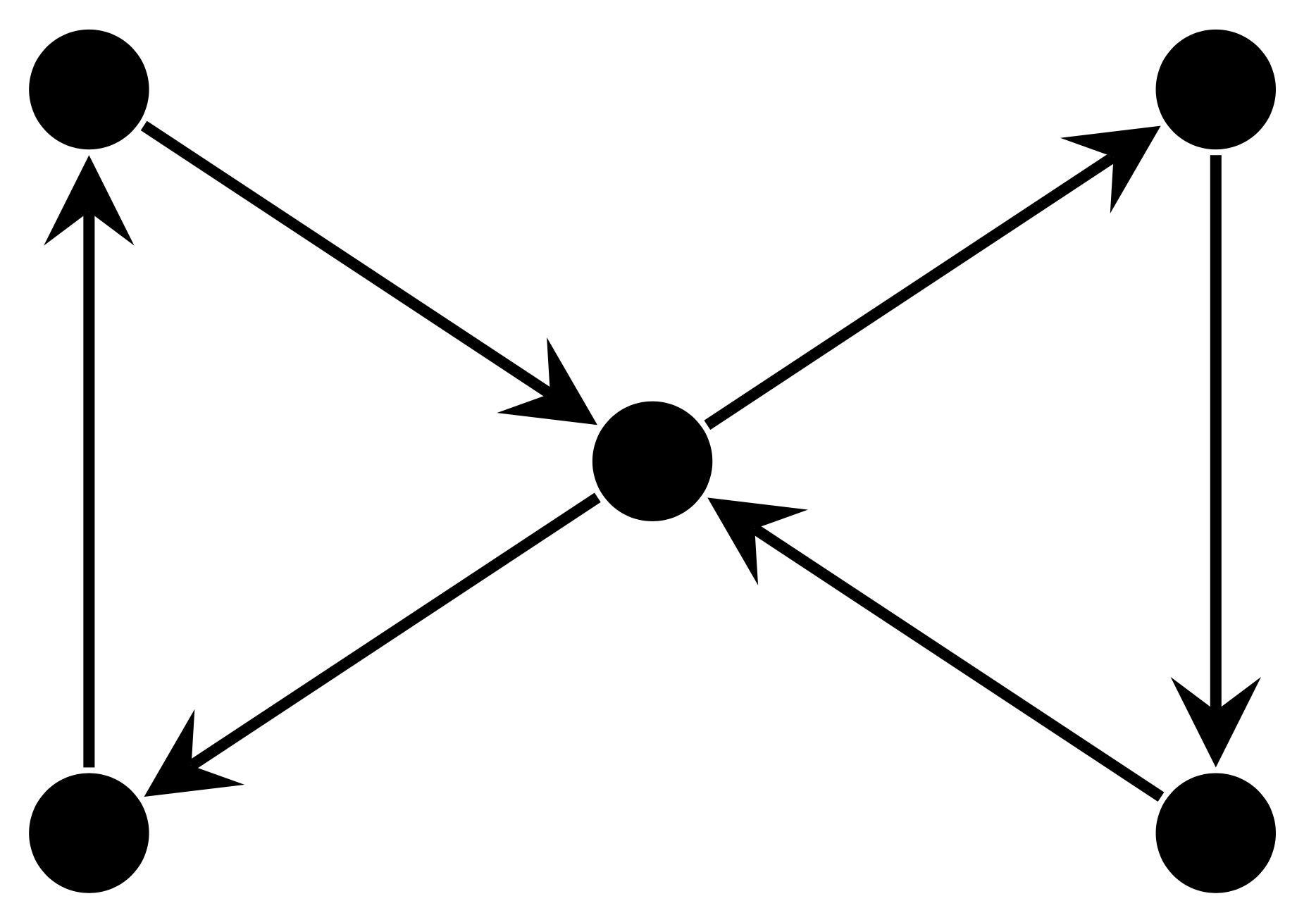}\!\!\ .
\]

Still, in \emph{undirected} graphs, the shortest even cycle problem is well understood, though the arguments are more sophisticated.
The earliest polynomial-time algorithms use Edmond's minimum-weight perfect matching algorithm. 
Later, Monien \cite{Monien83} published an algorithm with running time $O(n^2\alpha(n))$, which was improved to $O(n^2)$ by Yuster and Zwick~\cite{YusterZ97}.
Alas, these algorithms are based on combinatorial properties of undirected graphs that do not hold in directed graphs.
Thus, no algorithms for even cycles in \emph{directed} graphs follow from this work.

Nor was it clear that this problem should be tractable.

In fact, it was open for a long time whether there exists a polynomial-time algorithm for the \emph{recognition} version, the  Even Cycle problem:
``Given a directed graph, does it have an even cycle (no matter its length)?''
The question goes back to Younger~\cite{Younger73} in the early 1970s and 
was reiterated in many subsequent papers (\emph{e.g.},~\cite{
  ChungGK94, 
  Thomassen85, 
  Thomassen92,
  Thomassen93, 
  VaziraniY1989, 
  YusterZ97}).
Finally, at the turn of the millennium, McCuaig~\cite{McCuaig04} and independently Robertson, Seymour, and Thomas~\cite{RST99} gave a characterisation of undirected bipartite graphs meeting the requirements in P\'olya's permanent problem, which indirectly using already known reductions by Little~\cite{Little1975} and Seymour and Thomassen~\cite{SeymourT1987} lead to a polynomial time algorithm for the Even Cycle problem.
However, it is not at all clear how to use such a recognition-oracle to find a \emph{shortest} even cycle in a directed graph.

\medskip
Thus, neither the elementary algorithm for the odd case, nor the more sophisticated algorithms for the undirected case, nor the very extensive machinery behind the Even Cycle problem have led to an efficient algorithm for the \emph{Shortest} Even Cycle problem. As our main result, we present such an algorithm using a very different approach. 
\begin{Thm}[Computing the length of a shortest even cycle]
\label{thm:main}
  Given a directed graph with $n$ vertices, the length of a shortest even cycle can be found in time $\tilde O(n^{3+\omega})$ with probability at least $1-O(n^{-1})$.
  Here, $\omega$ is the square matrix multiplication exponent.
\end{Thm}

A shortest even cycle can thus be \emph{found} using standard self-reduction to 
the above result in time $\tilde O(n^{4+\omega})$. Our approach seems to be useful also for the Even Cycle problem, in particular we get a faster algorithm than was previously known for bounded genus graphs (cf.~\S\ref{sect:nested}). This latter algorithm can also be modified to solve the Shortest Even Cycle problem in this graph class faster than the general algorithm in Theorem~\ref{thm:main}.

\subsection{Overview of techniques}

Our paper is largely self-contained, and the correctness and running time arguments are quite short, certainly in comparison to \cite{McCuaig04, RST99}.
At a high level, our approach is to rely on algebraic fingerprinting~\cite{KoutisW2016} and a combination of the permanent and the determinant functions---implicit in a paper of Vazirani and Yannakakis~\cite{VaziraniY1989}---to obtain an edge-weighted enumeration of the cycle covers of the input graph by the parity of the number of cycles in the cover, which gives us fingerprinting-control on cycle covers containing a shortest even cycle. What makes this approach towards tractability nontrivial is the need to work with the permanent function, which is known to be \#P-hard except for very restricted families of rings by the work of Valiant~\cite{Valiant1979}; the most notable such family admitting efficient permanent computation are the integers modulo even prime powers. As such, our key technical contribution here amounts to engaging in ``designer commutative algebra'' to design a family of finite rings that 
simultaneously 
\begin{itemize}
\item[(i)] 
for a weighted adjacency matrix $A$,
give a nondegenerate representation for the parity cycle cover 
identity (cf.~\eqref{eq:pcc-to-per-det} in \S\ref{sect:pcc})
\[
2\pcc_{n-1} A=\per A-\det A\,;
\]
\item[(ii)] 
admit efficient permanent computation by extension of Valiant's 
techniques~\cite{Valiant1979}; and 
\item[(iii)] 
are sufficiently ``field-like'' to enable and benefit from standard design techniques for finite fields, in particular randomized polynomial identity 
testing~\cite{DeMilloL1978,Schwartz1980,Zippel1979}, polynomial interpolation,
and fast algorithms for determinants~\cite{BunchH1974,LabahnNZ2017} 
over finite fields.
\end{itemize}
Somewhat more precisely, for a finite field $\F_{2^d}$ of characteristic $2$, 
our design ``extends the characteristic'' to obtain a finite ring $\E_{4^d}$ 
of characteristic $4$ that satisfies the constraints above; crucially, 
whenever a (multivariate) polynomial identity $2p=q$ holds for two 
polynomials $p$ and $q$ over $\E_{4^d}$, with $p$ restricted 
to $\{0,1\}$-coefficients, we can {\em emulate} the evaluation 
of $p$ {\em over the finite field $\F_{2^d}$} by evaluating $q$ 
{\em over $\E_{4^d}$ instead}, followed by a simple inversion operation 
to recover the value of $p$ over $\F_{2^d}$. In our main application, $p$ is 
the parity cycle cover enumeration (which we evaluate over $\F_{2^d}$ using 
emulation), and $q$ is the difference of a permanent and a determinant, 
both of which admit fast algorithms over $\E_{4^d}$ using a 
{\em reverse-emulation} approach to run much of the computations 
in $\F_{2^d}$ using dedicated finite-field algorithms. We expect these design 
techniques to be potentially useful in other contexts as well. 

\medskip
\noindent
{\em Remark.} The present strategy of ``designing the ring'' has been 
foreshadowed in earlier work, in particular 
Bj\"orklund and Husfeldt~\cite{BjorklundH19} rely on permanent computations 
over a degree-truncated polynomial ring $\Z_{4}[x]/\bra x^d\ket$ 
of characteristic $4$ to obtain a randomized polynomial-time algorithm for 
the Shortest Two Disjoint Paths problem in undirected graphs.
This earlier 
design, however, does not support field-emulation and needs $d$ to be of size polynomial in $n$ as opposed to logarithmic in $n$ as we use here. This accordingly leads to 
considerably less efficient algorithms. For example, applying the present 
techniques, we can improve Shortest Two Disjoint Paths on an $n$-vertex, 
$m$-edge input from $O(n^{10}m^3)$ time in \cite{BjorklundH19} to 
$\tilde O(n^{3+\omega})$, cf.~\S\ref{sect:s2dp}.

The other direction---using the algebraic tools from \cite{BjorklundH19} to compute the expressions in the present paper---would also work, though we have not spelt out the details.
The resulting running time would be similar to that of \cite{BjorklundH19}.

\subsection{Related work}
\label{sect:related-work}

Let us now proceed to a more detailed discussion of related work.

\medskip
\noindent
{\em Girth and odd cycles.}
Algorithms for finding a shortest cycle in a graph (known as computing its girth) are textbook material and go back to Itai and Rodeh~\cite{ItaiR78}.
They are based on iterated breadth-first search in time $O(nm)$.
As mentioned above, these algorithms also find a shortest odd cycle because one of the bfs-trees contains a shortest closed odd walk, which must be simple.

\begin{table}
  \sf
\begin{tabular}{lll}\toprule
  Question & Time & Remarks \\\midrule
  Length of shortest cycle, girth & $O(nm)$ & BFS from every vertex; Itai and Rodeh~\cite{ItaiR78}\\
  \addlinespace[0.5ex]
  Does the graph contain an odd cycle? & $O(n+m)$ & DFS \\
  \addlinespace[0.5ex]
  Length of shortest odd cycle & $O(nm)$  & BFS from every vertex \\
  \addlinespace[0.5ex]
  Does the graph contain an even cycle? \\
  \hspace*{4cm}undirected & $O(n+m)$ & DFS; Arkin, Papadimitriou, and Yannakakis \cite{ArkinPY91} \\
  \hspace*{4cm}directed & $O(n^3)$ & Robertson, Seymour, and Thomas \cite{RST99}\\
  \addlinespace[0.5ex]
  Length of shortest even cycle \\
  \hspace*{4cm}undirected & $O(n^2)$ & Yuster and Zwick \cite{YusterZ97} \\
  \hspace*{4cm}directed & $\tilde O(n^{3+\omega})$ & This paper \\\bottomrule
\end{tabular}
  \vspace*{1mm}
  \caption{Overview of algorithms to find cycles and short cycles.}
\end{table}

The recognition problem ``Given a graph (directed or undirected), does it \emph{contain} an odd cycle'' is easier:
An undirected graph contains an odd cycle if and only if it is not bipartite.
Maybe less obviously, a directed graph contains an odd cycle if and only if one of its strongly connected components is non-bipartite .
Thus, odd cycle in a graph can be detected in time $O(n+m)$ using depth-first search.
All of these algorithm can be modified to output the cycle in question in the same time bound.

\medskip
\noindent
{\em Even cycles in undirected graphs.}
An undirected graph does not contain an even length cycle if and only if every biconnected component is an edge or an odd-length cycle, so the recognition problem is again solved in time $O(n+m)$ using depth-first search \cite{ArkinPY91,Thomassen88}.

It was also realised quite early how to find a \emph{shortest} even cycle in undirected graphs in polynomial time.
Early constructions are based on minimum perfect matchings; Thomassen~\cite{Thomassen85} attributes this argument to Edmonds, Monien~\cite{Monien83} to Gr\"otschel and Pulleyblank, who themselves credit ``Waterloo-folklore''~\cite{GrotschelP81}.
Monien then gave a sophisticated and much faster algorithm with running time $O(n^2\alpha(n))$ for finding a shortest even path.
This result was later improved by Yuster and Zwick to time $O(n^2)$~\cite{YusterZ97}. These algorithms are based on a variant of breadth-first search and use the fact that in an undirected graph, every shortest even cycle consists of two paths that are ``almost shortest paths'', \emph{i.e.}, they are at most one edge longer than a shortest path.

\medskip
\noindent
{\em Recognising even cycles in directed graphs.}
The history of the Even Cycle problem is rich, see McCuaig~\cite{McCuaig00} for a survey.
The problem has many equivalent characterisations, twenty-three of which are enumerated in McCuaig's systematic account~\cite{McCuaig04}.
For instance, is a given hypergraph with $n$ vertices and $n$ hyperedges minimally nonbipartite?
Does a given bipartite graph admit a Pfaffian orientation?
Is a given square matrix sign-nonsingular, i.e., is every matrix with the same sign pattern (plusses, minuses and $0$s in the same positions) nonsingular?
Perhaps the most famous version is P\'olya's permanent problem:
Given a $0$-$1$ matrix $A$, can you flip some of the $1$s to $-1$s creating another matrix $B$ so that $\per A  = \det B$?
Most important for the present paper is the characterisation of Vazirani and Yannakakis~\cite{VaziraniY1989}:
Given a square matrix $A$ of nonnegative integers, determine if $\det A = \per A$.

Polynomial-time algorithms for the Even Cycle problem were found independently by McCuaig~\cite{McCuaig04} and Robertson, Seymour, and Thomas~\cite{RST99}, announced jointly in~\cite{McCuaigRST97}.
The algorithm of~\cite{RST99} runs in time $O(n^3)$.
McCuaig eschews analysing the running time of the construction in~\cite{McCuaig04} and merely observes that it is polynomial.
An earlier paper of Thomassen~\cite{Thomassen93} showed that the Even Cycle problem in planar graphs could be solved in time $O(n^6)$. 

\medskip
\noindent
{\em Hardness results in directed graphs.}
Even though the Even Cycle problem in directed graphs admits a polynomial-time recognition algorithm, it seems difficult to extract any kind of information about length-constrained cycles in directed graphs.
For instance, it is NP-hard to determine if a directed graph contains 
(i) an odd cycle through a given edge~\cite{Thomassen85}, (ii) an even cycle through a given edge~\cite{Thomassen85},
  (iii) an odd chordless cycle~\cite{Lubiw88}, and
  (iv) an even chordless cycle~\cite{Lubiw88}.

It is also difficult to find \emph{balanced} cycles in the following sense.
In a directed graph in which each arc is colored in one of two colors, it is NP-hard to find (a necessarily even) cycle that alternates between the two colors \cite{GutinSY1998}. It is also NP-hard to find an even cycle that uses equally many arcs of each color. This can be observed by reducing from the NP-hard Hamiltonian path problem. For an $n$-vertex directed graph $G$, and two vertices $s$ and $t$, we want to detect if there is a Hamiltonian path from $s$ to $t$ in $G$. We construct a larger arc-colored graph $G'$ by copying $G$ and let each of its arcs get the first color. We next add $n-2$ vertices to $G'$ and connect them on a long directed path from $t$ to $s$ and give each of these arcs the second color. Then $G'$ has a color-balanced cycle iff $G$ has a $s\rightarrow t$ Hamiltonian path.

Finally, it is also hard to find a cycle whose length meets other remainder criteria.
For any modulus $m > 2$ and nonzero remainder $r$ with $0<r<m$ it is NP-hard to determine if a directed graph contains a cycle of length $r$ modulo $m$~\cite{ArkinPY91}.
The complexity of the case $m>2$, $r=0$ seems to be open~\cite{HemaspaandraST04}.

\section{Preliminaries}

This section outlines the key preliminaries for our algebraic fingerprinting
approach and develops the key connection to matrix permanent and matrix 
determinant. For general background on algebraic fingerprinting, 
cf.~e.g.~Koutis and Williams~\cite{KoutisW2016}.

\subsection{Commutative algebra}
\label{sect:poly}

We assume familiarity with elementary concepts in commutative algebra as well as
with the standard algorithmic toolbox for working with polynomials in
one indeterminate 
(cf.~e.g.~von zur Gathen and Gerhard~\cite{vonzurGathenG2013}). 

All rings in this paper are nontrivial ($0\neq 1$) and commutative without
further mention. For a ring $R$ and indeterminates $x_1,x_2,\ldots,x_n$,
we write $R[x_1,x_2,\ldots,x_n]$ for the ring of polynomials in the 
indeterminates $x_1,x_2,\ldots,x_n$ and with coefficients in $R$.
For a polynomial $p\in R[x_1,x_2,\ldots,x_n]$ and values 
$\xi_1,\xi_2,\ldots,\xi_n\in R$, we write 
$p(\xi)=p(\xi_1,\xi_2,\ldots,\xi_n)\in R$ 
for the evaluation of $p$ at $x_i=\xi_i$ for all $i=1,2,\ldots,m$. 
We use symbols from the Roman alphabet to denote polynomials and 
symbols from the Greek alphabet to denote elements of a ring of coefficients.
For an integer $m\geq 2$, we write $\Z_m$ for the ring of integers modulo $m$.

For a ring $R$ and ideal $I\subseteq R$, we write $R/I$ for the quotient ring
$R$ modulo $I$. For a generator $g\in R$, we write $\bra g\ket$ for the 
ideal generated by $g$. In this paper, we work only with quotient rings 
of the form $S[x]/\bra g\ket$, where $S$ is a coefficient ring and
$g\in S[x]$ is a generator polynomial of degree $d\geq 1$. 
In particular, basic arithmetic (addition, subtraction, multiplication, 
and---when available---multiplicative inverses) in $S[x]/\bra g\ket$ can 
be implemented using $\tilde O(d)$ black-box oracle calls for arithmetic 
in $S$ and the standard algorithmic toolbox for polynomials in one 
indeterminate, cf.~von zur Gathen and Gerhard~\cite{vonzurGathenG2013}.
(In this paper we will work only with the constant-size coefficient
rings $S=\Z_2$ and $S=\Z_4$; the standard toolbox thus gives us 
tacit $\tilde O(d)$-time arithmetic in $S[x]/\bra g\ket$.)

For a positive integer $d$, we write $\F_{2^d}$ for the finite field of 
order $2^d$ and assume this field is represented for purposes of arithmetic
as $\F_{2^d}=\Z_2[x]/\bra g_2\ket$, where $g_2\in\Z_2[x]$ is 
a $\Z_2$-irreducible polynomial of degree $d$. Given $d$ as input, 
we recall that we can construct such a polynomial $g_2$ in expected
time $\tilde O(d^2)$ as a one-off preprocessing step for all subsequent
arithmetic~(cf.~\cite[\S14.9]{vonzurGathenG2013}).

\subsection{Cycle covers}

Let $G$ be an $n$-vertex simple directed graph with a loop at every vertex.
We write $V(G)$ for the vertex set of $G$ and $E(G)$ for the arc set of $G$. 
A {\em cycle cover} of $G$ is a subset $C\subseteq E(G)$ such that for every 
vertex $u\in V$ there is exactly one arc leading into $u$ and exactly one 
arc leading out of $u$ in $C$; these two arcs are identical exactly when 
the arc is a loop. Thus, viewing each loop in $C$ as a cycle, we have that 
$C$ consists of exactly $n$ arcs which partition into vertex-disjoint 
directed cycles. Equivalently, we may view $C$ as a permutation of $V(G)$ 
that takes each $u\in V(G)$ into the head of the arc leading out of $u$ in $C$.
We write $\kappa(C)$ for the number of cycles in $C$
and specifically $\lambda(C)$ for the number of loops in $C$.
Let us write $\mathscr{C}(G)$ for the set of all cycle covers of $G$.

\subsection{Enumerating cycle covers by parity}
\label{sect:pcc}

Let $R$ be a ring and associate an arc 
weight $w_{uv}\in R$ with every arc $uv\in E(G)$. Let $A\in R^{n\times n}$ 
be an $n\times n$ weighted adjacency matrix with rows and columns indexed 
by $V(G)$ such that the entry $A_{u,v}$ at row $u\in V(G)$ and 
column $v\in V(G)$ of $A$ is 
\begin{equation}
\label{eq:adj-mat}
A_{u,v}=\begin{cases}
w_{uv} & \text{if $uv\in E(G)$};\\
0      & \text{otherwise}.
\end{cases}
\end{equation}
For an integer~$m$, define the {\em parity cycle cover enumerator} 
with {\em parity}~$m$ by
\begin{equation}
\label{eq:pcc}
\pcc_m A\,=\sum_{\substack{C\in\mathscr{C}(G)\\\kappa(C)\equiv m\!\!\!\!\pmod 2}}\,\prod_{uv\in C}w_{uv}\,.
\end{equation}
It is well known that determinant and permanent of $A$ satisfy
\begin{align*}
\det A\,&=\sum_{C\in\mathscr{C}(G)}(-1)^{n-\kappa(C)}\prod_{uv\in C}w_{uv}\,,\\
\per A\,&=\sum_{C\in\mathscr{C}(G)}\,\prod_{uv\in C}w_{uv}\,,
\end{align*}
or, what is the same,
\begin{align*}
\det A&=\pcc_n A-\pcc_{n-1} A\,,\\
\per A&=\pcc_n A+\pcc_{n-1} A\,.
\end{align*}
In particular, we have
\begin{equation}
\label{eq:two-pcc}
2\pcc_{n-1} A=\per A-\det A\,.
\end{equation}
The formula \eqref{eq:two-pcc} will form the core of our algebraic 
fingerprinting approach. 

\subsection{Even cycles via cycle cover enumeration}

We continue to work over a ring $R$. 
We say that an enumerator is {\em identically zero} if it has the value
zero independently of the chosen arc weights. 

Vazirani and Yannakakis \cite{VaziraniY1989} essentially showed
(their Lemma 2.2 was for $\{0,1\}$-matrices) the following lemma;
we give a short proof for convenience of exposition.

\begin{Lem}[Existence of an even cycle]
\label{lem:even-cycle-existence}
The graph $G$ has an even cycle if and only if
$\pcc_{n-1} A$ is not identically zero.
\end{Lem}
\begin{Proof}
When $G$ has only odd cycles, every cycle cover $C$ has 
$\kappa(C)\equiv n\pmod 2$ and thus $\pcc_{n-1} A$ is identically zero. 
Conversely, when $G$ has an even cycle, it can be extended with 
loops---recall that we assume that there is a loop at every vertex 
of $G$---to obtain a cycle cover $C$ with $\kappa(C)\equiv n-1 \pmod 2$. 
Thus, $\pcc_{n-1} A$ is not identically zero.
\end{Proof}

We can access the shortest even cycle by the following standard technique
of polynomial extension. Extend the weighted adjacency matrix 
$A\in R^{n\times n}$ in \eqref{eq:adj-mat} to a matrix 
$A_y\in R[y]^{n\times n}$ over the polynomial ring $R[y]$ in the 
indeterminate $y$ by multiplying all diagonal elements (that is, all 
loop-arc weights) of $A$ with the indeterminate $y$. 
More precisely, the entry $(A_y)_{u,v}$ at row $u\in V(G)$ and 
column $v\in V(G)$ of $A_y$ is defined by
\begin{equation}
\label{eq:adj-mat-y}
(A_y)_{u,v}=\begin{cases}
yw_{uv}   & \text{if $u=v$};\\
w_{uv}    & \text{if $u\neq v$ and $uv\in E(G)$};\\
0         & \text{otherwise}.
\end{cases}
\end{equation}
For a polynomial $p\in R[y]$ and nonnegative integer $\ell$,
let us write $[y^\ell]p\in R$ for the coefficient of the degree-$\ell$ 
monomial in $p$. 

\begin{Lem}[Length of a shortest even cycle]
\label{lem:shortest-even-cycle}
The length of a shortest even cycle in $G$ equals
the smallest positive even $k$ such that $[y^{n-k}]\pcc_{n-1} A_y$ is not
identically zero.
\end{Lem}
\begin{Proof}
Recall that we write $\lambda(C)$ for the number of loops in a cycle
cover $C\in\mathscr{C}(G)$. From \eqref{eq:pcc} and~\eqref{eq:adj-mat-y}, 
we have
\begin{equation}
\label{eq:y-pcc}
\pcc_{n-1} A_y\,=\!\!\!\!\!\!\!
\sum_{\substack{C\in\mathscr{C}(G)\\\kappa(C)\equiv n-1\!\!\!\pmod 2}}\!\!\!\!\!\!\!\!y^{\lambda(C)}\prod_{uv\in C}w_{uv}\,.
\end{equation}
When $G$ has only odd cycles, every cycle cover $C$ has 
$\kappa(C)\equiv n\pmod 2$, and thus $\pcc_{n-1} A_y$ is identically zero. 
So suppose that $G$ has an even cycle. Let $H$ be a shortest even cycle 
of $G$, and let $\ell$ be its length. Adjoin $n-\ell$ loops to $H$ to obtain 
a cycle cover $C_H$ with $\lambda(C_H)=n-\ell$ and 
$\kappa(C_H)=n-\ell+1\equiv n-1 \pmod 2$, since $\ell$ is even. 
In particular, we observe that the cycle cover $C_H$ defines a term
\begin{equation}
\label{eq:shortest-even-cycle-term}
y^{\lambda(C_H)}\prod_{uv\in C_H}w_{uv}=
y^{n-\ell}\prod_{uv\in E(H)}w_{uv}\prod_{u\in V(G)\setminus V(H)}w_{uu}
\end{equation}
in the enumeration~\eqref{eq:y-pcc}. In particular, 
$[y^{n-\ell}]\pcc_{n-1}A_y$ is not identically zero. 
Finally, suppose that $[y^{n-k}]\pcc_{n-1}A_y$ is not identically zero for 
an even $k$ with $k\leq\ell$. From \eqref{eq:y-pcc} we thus have that $G$ 
has a cycle cover $C$ with $\kappa(C)\equiv n-1\pmod 2$ and 
$\lambda(C)=n-k$. Since $k$ is even, we have that $C$ must contain 
an even cycle of length at most $k$. Since $\ell$ is the length of 
a shortest even cycle in~$G$, we conclude that $k=\ell$.
\end{Proof}

\section{Parity cycle cover enumeration in characteristic two}
\label{sec: parity cycle cover enumeration}

This section develops our main technical contribution, an efficient
algorithm for computing $\pcc_{n-1} A$ over a finite field of characteristic 
two. More precisely, in what follows we assume that the finite field 
$\F_{2^d}$ of order $2^d$ is represented as $\F_{2^d}=\Z_2[x]/\bra g_2\ket$,
where $g_2\in\Z_2[x]$ is a $\Z_2$-irreducible polynomial of degree $d$.
(For background on finite fields, see Lidl and Niederreiter~\cite{LidlN1997}.)
Once this algorithm is available, our main result then follows by standard 
finite-field polynomial identity testing and 
Lemma~\ref{lem:shortest-even-cycle}.

\subsection{Extending to characteristic four}

The core of our approach is to emulate arithmetic in the characteristic-two 
field $\F_{2^d}$ using a (to be defined) extension $\E_{4^d}$ 
in characteristic {\em four}, which supports an efficient algorithm for 
the permanent by a variation of Valiant's algorithm for the 
permanent modulo an even prime power~\cite{Valiant1979}.

Let us now define the ring $\E_{4^d}$ precisely. Recall that we write
$g_2\in\Z_2[x]$ for the $\Z_2$-irreducible polynomial of degree $d$ underlying
$\F_{2^d}=\Z_2[x]/\bra g_2\ket$. For a polynomial $a\in \Z_2[x]$, let us write 
$\bar a\in\Z_4[x]$ for the polynomial obtained by mapping the
$\{0,1\}$-reduced coefficients of $a$ to $\Z_4$. We say $\bar a$ is the 
{\em lift} of~$a$. Set $g_4=\bar{g}_2$ and define 
$\E_{4^d}=\Z_4[x]/\bra g_4\ket$.

Let us now proceed to connect $\F_{2^d}$ and $\E_{4^d}$. Towards this end, 
for a polynomial $s\in\Z_4[x]$, let us write $\underline{s}\in\Z_2[x]$ for
the polynomial obtained by reducing each coefficient of $s$ modulo $2$. 
We say that $\underline{s}$ is the {\em projection} of $s$. Projection
is readily verified to be a ring homomorphism from $\Z_4[x]$ to $\Z_2[x]$. 
Furthermore, projection inverts lift; that is, we have 
$\underline{\overline{a}}=a$ for all $a\in\Z_2[x]$.

We adopt the notational convention of using symbols in the Greek alphabet
for elements of $\F_{2^d}$ and $\E_{4^d}$. We use symbols 
$\alpha,\beta,\gamma,\ldots$ early in the alphabet for elements 
of $\F_{2^d}$ and symbols $\sigma,\tau,\upsilon,\ldots$ late in the 
alphabet for elements of $\E_{4^d}$. 

We extend the lift and project maps from the base polynomial rings 
$\Z_2[x]$ and $\Z_4[x]$ to the polynomial quotient rings $\F_{2^d}=\Z_2[x]/\bra g_2\ket$ 
and $\E_{4^d}=\Z_4[x]/\bra g_4\ket$ as follows. For 
$\alpha=a+\bra g_2\ket\in\F_{2^d}$ represented by $a\in\Z_2[x]$, 
we define the {\em lift} of $\alpha$ by
$\overline{\alpha}=\overline{a_\text{r}}+\bra g_4\ket\in\E_{4^d}$, 
where $a_\text{r}$ is the remainder of the polynomial division of $a$ by $g_2$. 
For $\sigma=s+\bra g_4\ket\in\E_{4^d}$ represented by $s\in\Z_4[x]$,
we define the {\em projection} 
of $\sigma$ by $\underline{\sigma}=\underline{s}+\bra g_2\ket\in\F_{2^d}$.

\begin{Lem}[Lift and project]
\label{lem:lift-and-project}
Both the lift map $\alpha\mapsto\overline{\alpha}$ and the projection
map $\sigma\mapsto\underline{\sigma}$ are well defined. 
Moreover, the projection map is a ring homomorphism from $\E_{4^d}$ 
to $\F_{2^d}$. 
\end{Lem}
\begin{proof}
Well-definedness is immediate for the lift map since we first reduce the
polynomial representative $a$ to the remainder $a_\text{r}$ before lifting. 
To see that the projection map is well-defined, let $s,s'\in\Z_4[x]$ 
with $s-s'=qg_4$ for some polynomial $q\in\Z_4[x]$. 
Since $\underline{g_4}=\underline{\overline{g_2}}=g_2$ and 
projection is a ring homomorphism from $\Z_4[x]$ to $\Z_2[x]$, we have
$\underline{s}-\underline{s'}=\underline{s-s'}=\underline{qg_4}=\underline{q}\,\underline{g_4}=\underline{q}g_2$. 
That is, the result of projection is independent of the chosen 
representative $s$ for $\sigma$, and thus $\underline{\sigma}$ is well defined. 
To verify that projection is 
a ring homomorphism from $\E_{4^d}$ to $\F_{2^d}$, by well-definedness 
it is immediate that $\underline{0}=0$ and $\underline{1}=1$. Let 
$\sigma=s+\bra g_4\ket\in\E_{4^d}$ be represented by $s\in\Z_4[x]$ and
$\tau=t+\bra g_4\ket\in\E_{4^d}$ be represented by $t\in\Z_4[x]$.
By well-definedness and the fact that projection is a homomorphism
from $\Z_4[x]$ to $\Z_2[x]$, we have both 
$\underline{\sigma+\tau}=
\underline{s+t}+\bra g_2\ket=
\underline{s}+\underline{t}+\bra g_2\ket=
\underline{\sigma}+\underline{\tau}$ and
$\underline{\sigma\tau}=
\underline{st}+\bra g_2\ket=
\underline{s}\,\underline{t}+\bra g_2\ket=
\underline{\sigma}\,\underline{\tau}$.
\end{proof}

Lifting and projection now enable emulation of arithmetic as follows.

\begin{Lem}[Emulating $\F_{2^d}$-arithmetic in $\E_{4^d}$]
\label{lem:emulation}
Let $e\in\E_{4^d}[x_1,x_2,\ldots,x_m]$ be a polynomial and let
$\underline{e}\in\F_{2^d}[x_1,x_2,\ldots,x_m]$ be 
obtained by projecting all the coefficients of monomials of $e$.
Then, for all $\alpha_1,\alpha_2,\ldots,\alpha_m\in\F_{2^d}$, we have
\begin{equation}
\label{eq:emulation}
2\overline{\underline{e}(\alpha_1,\alpha_2,\ldots,\alpha_m)}=
2e(\overline{\alpha_1},\overline{\alpha_2},\ldots,\overline{\alpha_m})\,.
\end{equation}
\end{Lem}
\begin{Proof}
For flexibility in what follows, let us prove a slightly stronger reverse 
form of the identity~\eqref{eq:emulation}. Namely, we proceed to show that 
for all polynomials $e\in\E_{4^d}[x_1,x_2,\ldots,x_m]$ and all 
$\tau_1,\tau_2,\ldots,\tau_m\in\E_{4^d}$, we have 
\begin{equation}
\label{eq:reverse-emulation}
2e(\tau_1,\tau_2,\ldots,\tau_m)=
2\overline{\underline{e}(\underline{\tau_1},\underline{\tau_2},\ldots,\underline{\tau_m})}\,;
\end{equation}
then \eqref{eq:emulation} follows from \eqref{eq:reverse-emulation} by
setting $\tau_i=\overline{\alpha_i}$ and observing
that $\underline{\tau_i}=\alpha_i$ for all $i=1,2,\ldots,m$.

To establish \eqref{eq:reverse-emulation}, first observe that 
the project-and-lift-times-$2$ identity $2u=2\overline{\underline{u}}$ holds 
for all polynomials $u\in\Z_4[x]$; indeed, consider the coefficients of $u$ 
and observe the modulo-$4$ congruence $2z\equiv 2(z\bmod 2)$ for all 
integers $z$. Thus, because the projection and lift maps are well-defined 
(Lemma~\ref{lem:lift-and-project}), we have 
$2\upsilon=2\overline{\underline{\upsilon}}$ for all 
$\upsilon=u+\bra g_4\ket\in\E_{4^d}$, and 
for $\upsilon=e(\tau_1,\tau_2,\ldots,\tau_m)$ 
in particular. Finally, use the fact that the projection map is 
a homomorphism (Lemma~\ref{lem:lift-and-project}) on the sum of terms 
of $e$ evaluated at $x_i=\tau_i$ for all $i=1,2,\ldots,m$ to conclude that
$2e(\tau_1,\tau_2,\ldots,\tau_m)
=2\overline{\underline{e(\tau_1,\tau_2,\ldots,\tau_m)}}
=2\overline{\underline{e}(\underline{\tau_1},\underline{\tau_2},\ldots,\underline{\tau_m})}$.
\end{Proof}

Algorithmically, we rely on the standard toolbox 
for basic algebraic operations on univariate polynomials over a 
black-box ring (cf.~\S\ref{sect:poly}); in particular, this enables 
tacit $\tilde O(d)$-time arithmetic in $\F_{2^d}$ and $\E_{4^d}$ 
in what follows.

\subsection{Reduction to the permanent and determinant over $\E_{4^d}$.}

We are now ready for our first reduction.
Let $A\in \F_{2^d}^{n\times n}$ be an $n\times n$ matrix given to us as input. 
We seek to compute the parity cycle cover enumerator $\pcc_{n-1} A$ 
over $\F_{2^d}$. Let us write $\bar A\in\E_{4^d}^{n\times n}$ for 
the entrywise lift of~$A$.
Observing that~\eqref{eq:two-pcc} holds in particular over the polynomial 
ring $\E_{4^d}[w_{uv}:uv\in E]$, it follows immediately 
from~\eqref{eq:emulation} that we have the $\E_{4^d}$-identity
\begin{equation}
\label{eq:pcc-to-per-det}
2\,\overline{\pcc_{n-1} A}
=2\pcc_{n-1}\bar{A}
=\per\bar{A}-\det\bar{A}\,.
\end{equation}
That is, to compute $\pcc_{n-1} A$ over $\F_{2^d}$, 
by \eqref{eq:pcc-to-per-det} it suffices to 
compute $\per\bar A-\det\bar A$ over $\E_{4^d}$ and then
invert the lift-times-$2$ operation to recover $\pcc_{n-1} A$ in $\F_{2^d}$.

Thus, it now remains to compute permanents and determinants fast 
over $\E_{4^d}$.

\subsection{Computing the permanent over $\E_{4^d}$}
\label{sect:per-e}

Throughout this section we work over $\E_{4^d}$ and seek to compute the
permanent $\per M$ of a given $n\times n$ matrix $M\in\E_{4^d}$ with
entries $\sigma_{i,j}\in\E_{4^d}$ for all $i,j\in [n]$ with 
$[n]=\{1,2,\ldots,n\}$. 

It is convenient to start by recalling the standard Leibniz-style definition 
of the permanent. Let us write $S_n$ for the symmetric group of all 
permutations $f:[n]\rightarrow [n]$. The {\em permanent} of $M$ is
\begin{equation}
\label{eq:per}
\per M = \sum_{f\in S_n}\sigma_{1,f(1)}\sigma_{2,f(2)}\cdots \sigma_{n,f(n)}\,.
\end{equation}
Since addition distributes over multiplication, from \eqref{eq:per} 
it follows immediately that the permanent satisfies the branching 
row operation
\begin{equation}
\label{eq:per-row}
\per M = \per M_{i_1,i_2,\tau}' + \per M_{i_1,i_2,\tau}''
\end{equation}
for all rows $i_1,i_2\in [n]$ and scalars $\tau\in\E_{4^d}$, where
we write 
\begin{itemize}
\item[(\ref{eq:per-row}a)]
$M_{i_1,i_2,\tau}'$ for the matrix obtained from $M$ 
by subtracting $\tau$ times row $i_1$ from row $i_2$, and
\item[(\ref{eq:per-row}b)]
$M_{i_1,i_2,\tau}''$ for the matrix obtained from $M$
by replacing row $i_2$ with $\tau$ times row $i_1$.
\end{itemize}
We say that row $i_2$ is {\em similar} to row $i_1$ if 
there exists a scalar $\tau\in\E_{4^d}$ such that row $i_2$ equals $\tau$ 
times row $i_1$. In particular, row $i_2$ is similar to row $i_1$ 
in $M_{i_1,i_2,\tau}''$.

Valiant~\cite{Valiant1979} observed that if a matrix (with integer entries) 
has two pairs of similar rows, then its permanent is zero modulo $4$; 
this is because
each monomial in \eqref{eq:per} picks one element per row but for different
columns, and we can swap the columns to get an identical term. 
Thus, \eqref{eq:per-row} enables an elimination procedure analogous to 
Gaussian elimination but with recursive branching to the two branches 
$M_{i_1,i_2,\tau}'$ and $M_{i_1,i_2,\tau}''$ at every elimination step;
crucially, Valiant's observation gives control on the number of branches 
that must be considered since the $M_{i_1,i_2,\tau}''$-branch can be 
discarded whenever it has two similar rows. A direct implementation of
this strategy leads to Valiant's $\tilde O(n^5)$-time algorithm 
for the permanent modulo 4, and would in a straightforward manner 
lead to an $\tilde O(n^5d)$-time algorithm design over $\E_{4^d}$. 
Our goal here, however, is a faster design that benefits from reverse emulation
and altogether avoids recursion by reduction to determinants in $\F_{2^d}$
on the $M_{i_1,i_2,\tau}''$-branch.

Before describing our algorithm in more detail, let us introduce terminology 
for elimination in~$\E_{4^d}$. For an element $\sigma\in\E_{4^d}$, we say
that $\sigma$ is {\em even} if every coefficient of $\sigma$ is even; 
otherwise $\sigma$ is {\em odd}. We observe that (i) multiplying with
an even element always gives an even result; and (ii) the product of any two
even elements is zero. From (ii) we have that any product in \eqref{eq:per} 
is zero unless it contains at most one even term. This observation enables 
computing $\per M$ using successive row operations~\eqref{eq:per-row} 
to eliminate odd entries until the permanent becomes trivial to compute; 
the following lemma shows how to compute the coefficients $\tau$ for 
the row operations.

\begin{Lem}[Odd-elimination]
\label{lem:odd-elimination}
For all $\sigma,\upsilon\in\E_{4^d}$ with $\sigma$ odd, 
there exists a $\tau\in\E_{4^d}$ such that $\upsilon-\sigma\tau$ is even.
\end{Lem}
\begin{Proof}
Since $\sigma$ is odd the projection $\underline{\sigma}$ is 
nonzero and thus has a multiplicative inverse $\underline{\sigma}^{-1}$ 
in $\F_{2^d}$. 
Take $\tau=\overline{\underline{\sigma}^{-1}}\upsilon$ and observe that
$\underline{\upsilon-\sigma\tau}
=\underline{\upsilon}-\underline{\sigma}\,\underline{\tau}
=\underline{\upsilon}-\underline{\sigma}\,\underline{\sigma}^{-1}\,\underline{\upsilon}=0$ in $\F_{2^d}$.
That is, $\upsilon-\sigma\tau$ is even.
\end{Proof}

Our algorithm for the permanent $\per M$ over $\E_{4^d}$ is as follows.
Maintain a matrix $M$, initialized to the given input; also 
maintain an accumulator taking values in $\E_{4^d}$, initialized to zero. 
Assume initially all rows and columns of $M$ are unmarked. Require the 
invariant that each marked column contains exactly one odd entry, and 
the submatrix of marked rows and marked columns has exactly one odd entry 
in each row. While there remain unmarked columns with odd entries in 
unmarked rows, select one such entry $\sigma=\sigma_{i_1,j}$, which we
assume to lie at row $i_1\in[n]$ and column $j\in[n]$.
Use row-operations~\eqref{eq:per-row} with 
coefficients $\tau$ from Lemma~\ref{lem:odd-elimination} to eliminate all 
other, if any, odd entries $\upsilon=\sigma_{i_2,j}$ in column $j$,
observing by (i) and the invariant that these operations do not introduce 
new odd entries to any of the marked columns; also observe that each 
row-operation~\eqref{eq:per-row} creates two branches, $M_{i_1,i_2,\tau}'$ 
and $M_{i_1,i_2,\tau}''$, of the current matrix $M$---we implement each such 
operation by assigning $M\leftarrow M_{i_1,i_2,\tau}'$ and adding 
the permanent $\per M_{i_1,i_2,\tau}''$ (which we compute using a dedicated 
subroutine described in what follows) to the accumulator. 
Mark row $i_1$, mark column $j$, and iterate. When the iteration stops, 
from the invariant we observe that any remaining unmarked rows must consist 
of even entries only. There can be at most one such row, or otherwise the 
permanent is zero by (ii) and \eqref{eq:per}. Thus, $\per M$ is trivial 
to compute when the iteration stops  since at most one term (defined by the 
odd entries and the entry at the intersection of the unmarked row and 
unmarked column, if any) in \eqref{eq:per-row} is nonzero. Add $\per M$ to 
the accumulator. Return the value of the accumulator and stop.

To process the $M_{i_1,i_2,\tau}''$-branches, we rely on the fact that
rows $i_1$ and $i_2$ in $M_{i_1,i_2,\tau}''$ are similar to reduce the 
task of computing $\per M_{i_1,i_2,\tau}''$ over $\E_{4^d}$ 
to a {\em determinant} computation over a univariate 
polynomial ring $\F_{2^d}[r]$. In essence, we rely on reverse emulation
enabled by similarity. 
\begin{Lem}[Permanent with a similar pair of rows reduces to determinant]
\label{lem:simpair-per-det}
Suppose the rows $i_1$ and $i_2$ in $M\in\E_{4^d}^{n\times n}$ are
similar with $i_1\neq i_2$. Let $B\in\F_{2^d}[r]^{n\times n}$ be 
obtained from the entrywise projection $\underline{M}\in\F_{2^d}^{n\times n}$ 
by
\begin{itemize}
\item[(i)]
multiplying row $i_1$ entrywise with the monomial vector $(1,r,r^2,\ldots,r^{n-1})$; and 
\item[(ii)]
multiplying row $i_2$ entrywise with the monomial vector $(r^{n-1},r^{n-2},\ldots,1)$.
\end{itemize}
Then, $\per M=2\overline{\sum_{\ell=0}^{n-2}[r^\ell]\det B}$.
\end{Lem}
\begin{Proof}
Let us study the permanent $\per M$ over $\E_{4^d}$ using \eqref{eq:per}.
Select an arbitrary $f\in S_n$ and study the monomial defined by $f$ 
in \eqref{eq:per}. Suppose that $f(i_1)=j_1$ and $f(i_2)=j_2$. 
Since $i_1\neq i_2$ and $f$ is a permutation, we have $j_1\neq j_2$. 
Define $f':[n]\rightarrow[n]$ for all $i\in [n]$ by
\begin{equation}
\label{eq:f-prime}
f'(i)=\begin{cases}
j_2 & \text{if $i=i_1$;}\\
j_1 & \text{if $i=i_2$; and}\\
f(i) & \text{otherwise.}
\end{cases}
\end{equation}
Observe that $f'\neq f$ is a permutation of $[n]$, and furthermore 
$(f')'=f$; that is, the map $f\mapsto f'$ is an involution that 
partitions $S_n$ into disjoint pairs $\{f,f'\}$ of permutations; 
form the subset $S_n'\subseteq S_n$ by selecting from each such pair
the permutation $g\in\{f,f'\}$ with $g(i_1)<g(i_2)$.
Furthermore, since rows $i_1$ and $i_2$ are similar 
in $M$, for all permutations $f\in S_n$ we have
\begin{equation}
\label{eq:sim-monomial}
\sigma_{1,f(1)}\sigma_{2,f(2)}\cdots \sigma_{n,f(n)}=
\sigma_{1,f'(1)}\sigma_{2,f'(2)}\cdots \sigma_{n,f'(n)}\,.
\end{equation}
Thus, from \eqref{eq:per} and \eqref{eq:sim-monomial} we have
\begin{equation}
\label{eq:sim-per}
\per M=2\sum_{f\in S_n'}\sigma_{1,f(1)}\sigma_{2,f(2)}\cdots \sigma_{n,f(n)}\,.
\end{equation}
From the reverse emulation identity \eqref{eq:reverse-emulation} applied to 
the right-hand side of~\eqref{eq:sim-per} it follows that 
to complete the proof it remains to show that over $\F_{2^d}$ we have
\begin{equation}
\label{eq:r-extraction}
\sum_{f\in S_n'}\underline{\sigma_{1,f(1)}}\,\underline{\sigma_{2,f(2)}}\cdots\underline{\sigma_{n,f(n)}}
=\sum_{\ell=0}^{n-2}[r^\ell]\per B
=\sum_{\ell=0}^{n-2}[r^\ell]\det B\,.
\end{equation}
The second equality in \eqref{eq:r-extraction} is immediate since 
determinant and permanent are equal in characteristic 2. To establish the
first equality in \eqref{eq:r-extraction}, let us write $b_{i,j}\in\F_{2^d}[r]$ 
for the entry of $B$ at row $i\in [n]$, column $j\in [n]$. Observe that (i) 
and (ii) imply that  for all $f\in S_n$ we have
\[
b_{i_1,f(i_1)}b_{i_2,f(i_2)}
=\underline{\sigma}_{i_1,f(i_1)}\underline{\sigma}_{i_2,f(i_2)}r^{n-1+f(i_1)-f(i_2)}\,.
\]
In particular, by the construction of $S_n'$ 
we have $n-1+f(i_1)-f(i_2)\leq n-2$ if and only if $f\in S_n'$,
and the first equality in \eqref{eq:r-extraction} thus follows.
\end{Proof}

Lemma~\ref{lem:simpair-per-det} in particular enables us to compute 
$\per M_{i_1,i_2,\tau}''$ in time $\tilde O(n^{\omega}d)$ via 
the Labahn-Neiger-Zhou algorithm.

\begin{Thm}[Labahn, Neiger, and Zhou~{\cite[Theorem~1.1]{LabahnNZ2017}}]
\label{thm:lnz}
Let\/ $\F$ be a finite field and let $B$ be a nonsingular matrix
in $\F[r]^{n\times n}$. There is a deterministic algorithm that
computes $\det B\in\F[r]$ using $\tilde O(n^{\omega}\lceil \mu\rceil)$ 
operations in~$\F$, with $\mu$ being the minimum of the average of
the degrees of the columns of $B$ and that of its rows. 
\end{Thm}

In applying Theorem~\ref{thm:lnz}, we need to verify nonsingularity; that is,
that $\det B$ is a nonzero polynomial in the indeterminate $r$. Since 
$\det B$ has degree at most $2n-2$ in $r$, we can select a uniform random
$\rho\in\F_{2^d}$, substitute $r=\rho$ in $B$ to obtain the matrix 
$B(\rho)\in\F_{2^d}^{n\times n}$, and compute $\det B(\rho)$ in 
time $\tilde O(n^\omega d)$ using the algorithm of 
Bunch and Hopcroft~\cite{BunchH1974}.
If $\det B(\rho)\neq 0$, then $B$ is nonsingular and we apply 
Theorem~\ref{thm:lnz} to determine $\det B$. 
If $\det B(\rho)=0$, then we assert that $\det B$ is the zero polynomial
and proceed accordingly. Since a nonzero univariate polynomial of degree 
$\Delta$ has at most $\Delta$ roots, we incorrectly assert that $\det B$ 
is zero with probability at most $2^{1-d}n$.

We conclude that each row operation can thus be implemented in 
$\tilde O(nd+n^\omega d)$ time, with a failure probability of at most
$2^{1-d}n$. Observing that there are at most $n^2$ row operations, 
and taking the union bound over the failure probabilities of each 
operation, we have our main result for the permanent over $\E_{4^d}$:
\begin{Lem}[Permanent over $\E_{4^d}$]
\label{lem:per-e}
There is a randomized algorithm that correctly computes the permanent of 
a given matrix $M\in\E_{4^d}^{n\times n}$ in $\tilde O(n^{2+\omega}d)$ 
time and with probability at least\/ $1-2^{1-d}n^3$.
\end{Lem}

\subsection{Computing the determinant over $\E_{4^d}$}
\label{sect:det-e}

To evaluate the right-hand side of \eqref{eq:pcc-to-per-det} fast, we still
need an algorithm that computes the determinant of a given
matrix $M\in\E_{4^d}^{n\times n}$. This can be accomplished, for example,
in time $\tilde O(n^4d)$ using the division-free determinant algorithm of 
Berkowitz~\cite{Berkowitz1984} over $\E_{4^d}^{n\times n}$. The
asymptotically fastest division-free algorithm due to 
Kaltofen~\cite{Kaltofen1992} run in time $\tilde O(n^{\omega/2+2}d)$
over $\E_{4^d}^{n\times n}$. Both of these algorithms work over an arbitrary 
commutative ring, and it turns out we can obtain a slightly faster design
tailored for the ring $\E_{4^d}$ by a slight modification of our permanent
algorithm in the previous section. Indeed, contrasting with the 
permanent~\eqref{eq:per} and recalling the standard Leibniz definition of the 
determinant
\begin{equation}
\label{eq:det}
\det M = \sum_{f\in S_n}(\sgn f)\sigma_{1,f(1)}\sigma_{2,f(2)}\cdots \sigma_{n,f(n)}\,,
\end{equation}
where we write $\sgn f\in\{-1,1\}$ for the sign of the permutation $f$, 
we observe that the analog of~\eqref{eq:per-row} for the determinant has 
the form 
\begin{equation}
\label{eq:det-row}
\det M = \det M_{i_1,i_2,\tau}'\,,
\end{equation}
in particular since the $M_{i_1,i_2,\tau}''$-branch always cancels for the
determinant due to $f$ and $f'$ having opposing signs for all $f\in S_n$. 
It thus follows we can use an iterative elimination procedure with row 
operations exactly as in the previous section to compute $\det M$, the only 
two modifications to the procedure being that (i) we always disregard 
the $M_{i_1,i_2,\tau}''$-branch since $\det M_{i_1,i_2,\tau}''=0$; 
and (ii) when the iteration stops, we compute the at most 
one {\em signed} term (defined by the odd entries and the entry at 
the intersection of the unmarked row and unmarked column, if any) and add
it to the accumulator. We thus have the following lemma for the determinant
over $\E_{4^d}$:

\begin{Lem}[Determinant over $\E_{4^d}$]
\label{lem:det-e}
There is an algorithm that computes the determinant of a given matrix
$M\in\E_{4^d}^{n\times n}$ in $\tilde O(n^3d)$ time.
\end{Lem}

\subsection{Parity cycle cover enumeration over $\F_{2^d}$}
\label{sect:pcc-ff}

Let us now summarize our main contribution in this section. Given as
input an $n\times n$ matrix $A\in\F_{2^d}^{n\times n}$, we have a randomized
algorithm that in time $\tilde O(n^{\omega+2}d)$ computes 
$\pcc_{n-1} A\in\F_{2^d}$. Indeed, from the given $A$ we first compute 
the entrywise lift $\bar A\in\E_{2^d}^{n\times n}$, then 
use Lemma~\ref{lem:per-e} to compute the permanent 
$\per\bar A\in\E_{4^d}$, then use Lemma~\ref{lem:det-e} to 
compute the determinant $\det\bar A\in\E_{4^d}$, 
then compute the difference $\per\bar A-\det\bar A\in\E_{4^d}$, and finally 
invert the lift-times-$2$ operation on the difference to recover 
by \eqref{eq:pcc-to-per-det} the
parity cycle cover enumeration $\pcc_{n-1} A\in\F_{2^d}$. 
We thus have:

\begin{Lem}[Parity cycle cover enumerator over $\F_{2^d}$]
\label{lem:pcc-ff}
There is a randomized algorithm that correctly computes the parity cycle 
cover enumerator $\pcc_{n-1}A$ of a given matrix $A\in\F_{2^d}^{n\times n}$ 
in $\tilde O(n^{2+\omega}d)$ time and with probability at least\/ $1-2^{1-d}n^3$.
\end{Lem}

As a concluding remark, let us observe that the elimination steps 
in \S\ref{sect:per-e} and \S\ref{sect:det-e} trace identical 
$M_{i_1,i_2,\tau}'$-branches towards the base case, and thus time savings
can be obtained in an implementation by accumulating both $\per M$ and 
$\det M$ simultaneously.

\section{An efficient randomized algorithm for shortest even cycle}

This section proves Theorem~\ref{thm:main}, relying on Lemma~\ref{lem:pcc-ff}
as the key subroutine. We start by developing well-known preliminaries 
in polynomial identity testing.

\subsection{Randomized polynomial identity testing}

We recall a squarefree variant of 
the DeMillo--Lipton--Schwartz--Zippel lemma~\cite{DeMilloL1978,Schwartz1980,Zippel1979}. Let $\F$ be a finite field. Let us write $|\F|$ for 
the order of $\F$.
We say that a monomial $w_1^{d_1}w_2^{d_2}\cdots w_m^{d_m}$ is
{\em squarefree} if $d_1,d_2,\ldots,d_m\in\{0,1\}$.
A polynomial $p\in\F[w_1,w_2,\ldots,w_m]$ is {\em squarefree} if all of 
its monomials are squarefree.

\begin{Lem}[Squarefree DeMillo--Lipton--Schwartz--Zippel]
\label{lem:squarefree-pit}
Let $p\in\F[w_1,w_2,\ldots,w_m]$ be a squarefree and nonzero polynomial 
of degree at most $\Delta$. Suppose that $\beta_1,\beta_2,\ldots,\beta_m\in\F$ 
are drawn independently and uniformly at random. Then, 
$p(\beta_1,\beta_2,\ldots,\beta_m)\neq 0$ with probability
at least $\bigl(1-\frac{1}{|\mathbb{F}|}\bigr)^\Delta$.
\end{Lem}
\begin{proof}
By induction on $\Delta$. The base case $\Delta=0$ is immediate. 
Let $\Delta\geq 1$. Since $p$ is squarefree, there exists an indeterminate 
$w_k$ and $p',p''\in\F[w_1,w_2,\ldots,w_{k-1},w_{k+1},\ldots,w_m]$ 
such that (i) $p=w_kp'+p''$ and (ii) $p'$ has degree at most $\Delta-1$. 
Let $\gamma=p''(\beta_1,\beta_2,\ldots,\beta_{k-1},\beta_{k+1},\ldots,\beta_m)$. By the induction hypothesis, 
$\eta=p'(\beta_1,\beta_2,\ldots,\beta_{k-1},\beta_{k+1},\ldots,\beta_m)\neq 0$ with probability at least 
$\bigl(1-\frac{1}{|\mathbb{F}|}\bigr)^{d-1}$. Conditioning on this event,
we have $p(\beta_1,\beta_2,\ldots,\beta_m)=\beta_k\eta+\gamma\neq 0$ 
if and only if $\beta_k\neq-\gamma\eta^{-1}$, which happens with 
probability $1-\frac{1}{|\mathbb{F}|}$ due to independence. 
Thus, $p(\beta_1,\beta_2,\ldots,\beta_m)\neq 0$ with probability at 
least $\bigl(1-\frac{1}{|\mathbb{F}|}\bigr)^\Delta$.
\end{proof}

\subsection{Algorithm for shortest even cycle}
\label{sec:secalg}

We are now ready for our main algorithm. Let us start by setting up the
algebraic context for the algorithm and only then give the algorithm in detail.

To set the context, let $G$ be an $n$-vertex simple directed graph with 
a loop at every vertex. Recall from Lemma~\ref{lem:shortest-even-cycle} that 
the smallest positive even $k$ such that $[y^{n-k}]\pcc_{n-1}A_y$ is not 
identically zero is the length of a shortest even cycle in $G$. 
Our algorithm witnesses such a value $k$, if any, with high probability 
by applying squarefree randomized polynomial identity testing 
(Lemma~\ref{lem:squarefree-pit}) to the polynomial 
$p=[y^{n-k}]\pcc_{n-1}A_y\in\F_{2^d}[w_{uv}:uv\in E(G)]$. 
That is, we choose the ring $R$ in 
Lemma~\ref{lem:shortest-even-cycle} to be the polynomial ring 
$\F_{2^d}[w_{uv}:uv\in E(G)]$, and choose the arc weights 
in \eqref{eq:adj-mat-y} to equal the indeterminates $w_{uv}$ of this 
polynomial ring. With these choices, $[y^{n-k}]\pcc_{n-1}A_y$ is not 
identically zero if and only if it is a nonzero polynomial, 
so Lemma~\ref{lem:squarefree-pit} applies. We would like to stress here that
the algorithm never works with the polynomial $p$ in a full explicit 
representation since this would be computationally too expensive; rather, 
the algorithm merely seeks to {\em witness that the polynomial is nonzero} 
by establishing that $p(\beta)=p(\beta_{uv}:uv\in E(G))\neq 0$ for an 
independent uniform random choice of values $\beta_{uv}\in\F_{2^d}$ 
for $uv\in E(G)$.

Let us now present the algorithm in detail. 
Let the given input be an $n$-vertex simple directed graph $G$ with 
a loop at every vertex. We may assume $n\geq 2$; indeed, otherwise $G$ 
has no even cycle. The algorithm tacitly relies on the standard 
algorithmic toolbox for univariate polynomials over a black-box 
ring to enable $\tilde O(d)$-time arithmetic operations 
in $\F_{2^d}$ (cf.~\S\ref{sect:poly}). 
\begin{itemize}
\item[(S1)]
Set $d\leftarrow 5\lceil\log_2 n\rceil$ and let 
$\gamma_0,\gamma_1,\ldots,\gamma_n\in\F_{2^d}$ be arbitrary distinct values.
\item[(S2)]
For each arc $uv\in E(G)$ independently, 
draw a uniform random value $\beta_{uv}\in\F_{2^d}$. 
\item[(S3)]
For each $\ell=0,1,\ldots,n$ in turn, compute 
$\delta_\ell\leftarrow\pcc_{n-1} A_{\gamma_\ell}(\beta)\in\F_{2^d}$ 
using the algorithm in Lemma~\ref{lem:pcc-ff} on the matrix 
$A_{\gamma_\ell}(\beta)\in\F_{2^d}^{n\times n}$ whose entry at each 
row $u\in V(G)$ and each column $v\in V(G)$ is defined by
\[
\bigl(A_{\gamma_\ell}(\beta)\bigr)_{u,v}=
\begin{cases}
\gamma_\ell\beta_{uu} & \text{if $u=v$};\\
\beta_{uv}            & \text{if $u\neq v$ and $uv\in E(G)$};\\
0                     & \text{otherwise}.
\end{cases}
\]
[Observe that we get $A_{\gamma_\ell}(\beta)$ by assigning 
$y\leftarrow\gamma_\ell$ and $w_{uv}\leftarrow\beta_{uv}$ 
for all $uv\in E(G)$ in \eqref{eq:adj-mat-y}.
We also observe that, for all possible outcomes of (S2), 
the probability for the bad event that at least one of the $n+1$ applications 
of the randomized algorithm in Lemma~\ref{lem:pcc-ff} fails is, by the union 
bound, at most $2^{2-d}n^4=O(n^{-1})$. Let us condition in what follows 
that the bad event does not happen.]
\item[(S4)]
    Determine the coefficients of the unique polynomial $q\in\F_{2^d}[y]$ of degree at most $n$ 
that satisfies $q(\gamma_\ell)=\delta_\ell$ for all $\ell=0,1,\ldots,n$.
For example, by Lagrange interpolation we have
\[
q=\sum_{\ell=0}^n\delta_\ell\prod_{\substack{j=0\\j\neq\ell}}^n \frac{y-\gamma_j}{\gamma_\ell-\gamma_j}\,.
\]
[Here we have $q=\pcc_{n-1}A_y(\beta)$ 
by (S3), \eqref{eq:adj-mat-y}, and \eqref{eq:y-pcc}.]
\item[(S5)]
Return the smallest positive even $k$ such that $[y^{n-k}]q\neq 0$; 
or, when no such $k$ exists, assert that $G$ has no even cycle.\\{}
[To analyse correctness, observe that when $G$ has no even cycle, 
we have $q=0$ and thus the algorithm will assert that $G$ has no even cycle. 
So let $k$ be the length of a shortest even cycle of $G$. 
Observing that (i) $p=[y^{n-k}]\pcc_{n-1}A_y$ has degree $n$ in the 
indeterminates $w_{uv}$ and (ii) $[y^{n-k}]q=p(\beta)$, 
from Lemma~\ref{lem:shortest-even-cycle} and 
Lemma~\ref{lem:squarefree-pit} we have that $[y^{n-k}]q=p(\beta)\neq 0$
with probability at least 
$\bigl(1-2^{-d}\bigr)^n\geq \bigl(1-n^{-5}\bigr)^n=1-O(n^{-4})$.
Thus, taking into account the conditioning of the bad event in (S3)
not happening, the algorithm succeeds with probability at least 
$1-O(n^{-1})$.]
\end{itemize}
We observe that the running time is dominated by (S3), which executes 
$n+1$ times the $\tilde O(n^{\omega+2}d)$-time algorithm 
in Lemma~\ref{lem:pcc-ff}.
Since $d=O(\log n)$, the running time of the algorithm is 
$\tilde O(n^{\omega+3})$. This completes the proof of Theorem~\ref{thm:main}.

\section{A faster randomized algorithm for detecting an even cycle}
\label{sect:nested}

This section develops a faster algorithm for the existence problem of
even cycles in bounded genus graphs, in particular planar graphs.
The algorithm is based on Lemma~\ref{lem:even-cycle-existence}, 
randomized polynomial identity testing, and more fine-grained {\em pivot-free} 
versions of the elimination procedures underlying Lemma~\ref{lem:pcc-ff}. 
In particular, we will rely on the technique of {\em nested dissection}, 
originally introduced by George~\cite{George1973} to obtain speed-up and 
space savings when solving systems of linear equations resulting from 
a 2-dimensional mesh, and generalized by Yuster~\cite{Yuster2008} and later by 
Alon and Yuster~\cite{AlonY2013} to matrices supporting pivot-free Gaussian
elimination over a finite field. 

\begin{Thm}[Even cycles in bounded genus graphs]
  \label{thm:boundedgenus}
  Given a directed graph $G$ of bounded genus with $n$ vertices, detecting whether $G$ has an even cycle or not can be done in time $\tilde O(n^{2+\frac{1}{2}})$ with probability at least $1-O(n^{-1})$. The length of a shortest even cycle can be found in time $\tilde O(n^{3+\frac{1}{2}})$.
\end{Thm}

Here we use the central random matrix perturbation idea of Alon and Yuster 
(Lemma~2.4 in~\cite{AlonY2013}) but in a new way that will enable 
pivot-freeness with high probability. We start by defining pivot-freeness 
in our context.

\subsection{Pivot-free elimination and the fill}
\label{sect:pivot-free}

Let us recall the gist of the elimination procedures in \S\ref{sect:per-e}
and \S\ref{sect:det-e}. Namely, we start with an $n\times n$ matrix 
$M\in\E_{4^d}^{n\times n}$ of initially unmarked rows and columns, and use 
row operations \eqref{eq:per-row} and \ref{eq:det-row} to expand the marked 
rows and columns, 
while maintaining the invariant that each marked column has exactly one odd 
entry, and the submatrix of marked rows and marked columns has exactly one 
odd entry in each row. Essential to this expansion is the selection of an 
odd {\em pivot} entry $\sigma=\sigma_{i_1,j}$ at an unmarked column $j\in[n]$ 
and unmarked row $i_1\in[n]$, which is then used in the row 
operations relative to other rows $i_2\in[n]$ to eliminate odd entries in 
column $j$ via Lemma~\ref{lem:odd-elimination}, after which the column $j$ 
and the row $i_1$ are both marked.

We say that the matrix $M$ admits {\em pivot-free} elimination if, during
elimination as above, we can always choose the pivot $\sigma=\sigma_{i_1,j}$ 
to be a diagonal entry with $i_1=j$.
For example, a triangular matrix with a diagonal of odd entries admits 
pivot-free elimination. Let us say that the {\em fill} is the number of matrix 
entries that are made nonzero at any point of the elimination process. 

In what follows we tacitly work with a sparse representation of all 
the matrices considered, that is, we represent an $n\times n$ matrix as a
list of tuples $(i,j,\sigma_{i,j})$ for all the nonzero 
entries $\sigma_{i,j}\neq 0$ with $i,j\in[n]$; furthermore, we tacitly assume
the list is indexed with appropriate data structures supporting 
$O(\log n)$-time access to rows and columns.

\subsection{Separators and nested dissection to control the fill}

Crucial to controlling the fill for a given matrix $M$ is the order in 
which the diagonal entries are processed. 
We say that an {\em undirected} graph $G$ with vertex set $V(G)=[n]$
{\em supports} the matrix $M$ if for all $i,j\in[n]$ it holds that
the entry $\sigma_{i,j}$ of $M$ is nonzero only if $\{i,j\}\in E(G)$.
Since $V(G)=[n]$, we observe that any ordering of the diagonal
elements of $M$ defines a unique ordering of the vertices of $G$
and vice versa.

To study the fill, we use graph separators as defined by 
Lipton and Tarjan~\cite{LiptonT1979}. 
We say that a class $\mathscr{C}$ of 
{\em undirected} graphs satisfies an $f(n)$-{\em separator theorem} for 
a function $f$ and constants $c<1$, $c'>0$, $n_0\geq 0$ if for every 
$n$-vertex graph $G$ in $\mathscr{C}$ with $n>n_0$ there exists 
a partition $A\cup B\cup C = V(G)$ with
\[
|A|\leq cn\,,
\quad
|B|\leq cn\,,
\quad
|C|\leq c'f(n)\,,
\]
and no edge joins a vertex of $A$ with a vertex of $B$ in $G$.
In particular, graphs of bounded genus satisfy a $n^{1/2}$-separator theorem, 
and one can in $O(n\log n)$-time find a so-called weak separator tree for 
any bounded genus graph, see Alon and Yuster~\cite{AlonY2013}. 

Given the weak separator tree, Gilbert and Tarjan~\cite{GilbertT1987} present
their Algorithm ND that labels the vertices of the graph according to 
a post-order traversal of the separator tree in time $O(n)$ so that the 
cuts get higher labels than the subgraphs they split. This enables us to 
control the fill of $M$ by running their Algorithm ND on a graph 
supporting $M$, and working with respect to the vertex order produced 
by the algorithm when executing elimination on $M$. The total running time 
of this reordering is $O(n\log n)$, as the time-dominant operation is to 
compute the separator tree. 

Our focus here is on bounded-genus graphs, but we observe that we could use 
the technique for other so-called $\delta$-sparse hereditary families of 
graphs, including ones that take longer to obtain 
a weak separator tree for, again see~\cite{AlonY2013} for some examples. 
Central to the efficiency of the method is the following bound on the fill 
that in particular applies to bounded-genus graphs:

\begin{Thm}[Gilbert and Tarjan~{\cite[Theorem 2]{GilbertT1987}}]
\label{thm:gilbert-tarjan}
Let $\mathscr{C}$ be a class of graphs that satisfy a $n^{1/2}$-separator
theorem and is closed under contraction and subgraph. Suppose that no
$n$-vertex graph in $\mathscr{C}$ has more than $\delta n+O(1)$ edges.
If $G$ in $\mathscr{C}$ has $n>n_0$ vertices, the ND order causes
$O(\delta n\log n)$ fill.
\end{Thm}

For our subsequent analysis of the branching elimination strategy that
we will pursue here, we will use the following slightly more precise structural
fact about Algorithm ND~\cite[Algorithm~2]{GilbertT1987} and 
the {\em ND order} it outputs: for an $n$-vertex undirected graph $G$ given 
as input with $n>n_0$, the top-level separator $C\subseteq V(G)$ satisfies
$|C|\leq c'n^{1/2}$ and splits the graph $G-C$ into $t$ connected components 
with vertex sets $A_1,A_2,\ldots,A_t\subseteq V(G)$ satisfying $|A_j|\leq cn$ 
for all $j=1,2,\ldots,t$; the algorithm then recurses on each of 
connected components $G[A_1],G[A_2],\ldots,G[A_t]$. The ND order output by the
algorithm satisfies $A_1<A_2<\cdots A_t<C$. In particular, vertices in $C$
are eliminated last.

\subsection{A randomized algorithm design}
\label{sect:rand-even-cycle}

We are now ready for our main algorithm design in this section. 
Again it is convenient to first set up the algebraic context for the 
algorithm and only then give the algorithm in detail.
We will postpone the description and analysis of the 
fine-grained elimination subroutine to the next subsection.

To set the context, let $G$ be an $n$-vertex simple directed graph with 
a loop at every vertex.
Our task is to decide whether $G$ has an even cycle.
Recall from Lemma~\ref{lem:even-cycle-existence} that 
$\pcc_{n-1}A$ is not identically zero if and only if $G$ has an even cycle.
Our algorithm witnesses that $\pcc_{n-1}A$ is not identically zero
with high probability by applying squarefree randomized polynomial identity 
testing (Lemma~\ref{lem:squarefree-pit}) to the polynomial 
$p=\pcc_{n-1}A\in\F_{2^d}[w_{uv}:uv\in E(G)]$. That is, we choose the ring $R$ 
in Lemma~\ref{lem:even-cycle-existence} to be the polynomial ring 
$\F_{2^d}[w_{uv}:uv\in E(G)]$, and choose the arc weights 
in \eqref{eq:adj-mat} to equal the indeterminates $w_{uv}$ of this 
polynomial ring. With these choices, $\pcc_{n-1}A$ is not 
identically zero if and only if it is a nonzero polynomial, 
so Lemma~\ref{lem:squarefree-pit} applies.

Let us now present the algorithm in detail. 
Let the given input be an $n$-vertex simple directed graph $G$ with 
a loop at every vertex. Suppose that the undirected graph underlying $G$
belongs to a graph class $\mathscr{C}$ that satisfies the assumptions of
Theorem~\ref{thm:gilbert-tarjan}. In particular, this applies to a 
graph of bounded genus; such graphs have bounded average degree $\delta$ 
(see e.g.~\cite{GilbertT1987}), which we will apply tacitly in what follows.
We may assume $n\geq 2$; indeed, otherwise $G$ 
has no even cycle. The algorithm tacitly relies on the standard 
algorithmic toolbox for univariate polynomials over a black-box 
ring to enable $\tilde O(d)$-time arithmetic operations 
in $\F_{2^d}$ (cf.~\S\ref{sect:poly}). 
\begin{itemize}
\item[(D1)]
Set $d\leftarrow 4\lceil\log_2 n\rceil$.
\item[(D2)]
For each arc $uv\in E(G)$ independently, 
draw a uniform random value $\beta_{uv}\in\F_{2^d}$. 
\item[(D3)]
Construct the matrix $A(\beta)\in\F_{2^d}^{n\times n}$ whose entry 
at each row $u\in V(G)$ and each column $v\in V(G)$ is defined by
\[
\bigl(A(\beta)\bigr)_{u,v}=
\begin{cases}
\beta_{uv}            & \text{if $uv\in E(G)$};\\
0                     & \text{otherwise}.
\end{cases}
\]
[Observe that we get $A(\beta)$ by assigning 
$w_{uv}\leftarrow\beta_{uv}$ for all $uv\in E(G)$ in \eqref{eq:adj-mat}.]
\item[(D4)]
Use the algorithm in Theorem~\ref{thm:gilbert-tarjan} on the undirected
graph underlying $G$ to compute an order for diagonal elements of 
$A(\beta)\in\F_{2^d}^{n\times n}$.\\{}
[Observe that the underlying undirected graph of $G$ supports $A(\beta)$.
This step takes time $O(n\log n)$.]
\item[(D5)]
Compute $\epsilon\leftarrow\pcc_{n-1} A(\beta)\in\F_{2^d}$ 
using \eqref{eq:pcc-to-per-det} 
and pivot-free elimination in the (D4) order to evaluate
$\per\overline{A(\beta)}$ and $\det\overline{A(\beta)}$.\\{}
[We postpone a detailed description and analysis of this subroutine to 
the next subsection. Here we will be content with observing that the failure
probability is at most $O(n^{-1})$ and the running time 
is $\tilde O(n^{2+\frac{1}{2}})$.]
\item[(D6)]
If $\epsilon\neq 0$, assert that $G$ contains an even cycle;
if $\epsilon=0$, assert that $G$ has no even cycle.\\{}
[To analyse correctness, observe that when $G$ has no even cycle, 
we have that $\pcc_{n-1} A$ is the zero polynomial by 
Lemma~\ref{lem:even-cycle-existence}, and hence $\pcc_{n-1} A(\beta)=0$
for all choices in (D2). Thus $\epsilon=0$ with probability at 
least $1-O(n^{-1})$ by the failure analysis in (D5). When $G$ has an even 
cycle, we have that $\pcc_{n-1} A$ is a nonzero polynomial of degree 
at most $n$ by Lemma~\ref{lem:even-cycle-existence} and \eqref{eq:pcc}. Thus, 
from Lemma~\ref{lem:squarefree-pit} we have that the 
probability for $\pcc_{n-1} A(\beta)=0$ is at most 
$1-\bigl(1-2^{-d}\bigr)^n\geq 1-\bigl(1-n^{-4}\bigr)^n=O(n^{-3})$.
Thus, by the union bound with the failure analysis in (D5), we have 
that $\epsilon\neq 0$ with probability at least $1-O(n^{-1})$.]
\end{itemize}

We observe that the running time is dominated by (D5), which runs in 
time $\tilde O(n^{2+\frac{1}{2}})$. 

By using the interpolation idea from Algorithm S in~\S\ref{sec:secalg}, 
and again using Lemma~\ref{lem:shortest-even-cycle} instead 
of Lemma~\ref{lem:even-cycle-existence}, we can solve for the length of 
a shortest even cycle in time $\tilde O(n^{3+\frac{1}{2}})$ in graphs 
of bounded genus. In more detail, 
\begin{enumerate}
\item[(i)] 
we replace step (D1) with (S1) but also require 
the $\gamma$-values to be non-zero,
\item[(ii)] 
we insert a loop for $\ell=0,1,\ldots,n$ as in (S3) immediately 
after step (D2),
\item[(iii)] we replace the diagonal entries of the matrix 
to $\bigl(A_{\gamma_\ell}(\beta)\bigr)_{u,u}=\gamma_\ell \beta_{uu}$ in 
step (D3),
\item[(iv)] 
we compute 
$\epsilon_\ell\leftarrow\pcc_{n-1} A_{\gamma_\ell}(\beta)\in\F_{2^d}$ in (D5), 
and
\item[(v)] 
we replace (D6) for steps (S4) and (S5) after exchanging $\delta_\ell$ 
by $\epsilon_\ell$.
\end{enumerate}
This completes the proof of Theorem~\ref{thm:boundedgenus}, pending the
detailed development of Step (D5) in the next section.

\subsection{Branching elimination over $\E_{4^d}$ in ND order}

This section develops a fine-grained elimination procedure for computing
$\per M$ and $\det M$ for a given matrix $M\in\E_{4^d}^{n\times n}$
supported by an undirected graph $G$ of bounded genus with $V(G)=[n]$. 
For convenience, we assume that both the matrix $M$ and the graph $G$ have 
been permuted to the ND order given by (D4); more precisely, we assume that 
the ND order for $M$ and $G$ is the natural numerical ordering 
$1,2,\ldots,n$ of $[n]$, where $1$ is eliminated first and $n$ last.
For $i=1,2,\ldots,n$, we always eliminate with the pivot-free diagonal 
choice $\sigma=\sigma_{i,i}$ (cf.~\S\ref{sect:pivot-free} 
and \S\ref{sect:per-e}), which we will show in what follows is odd 
for all the choices on all the branches considered with high probability. 

We focus on computing the permanent in what follows, with the understanding
that the determinant can be obtained with an analogous but simpler elimination
strategy since the determinant vanishes on the $M_{i_1,i_2,\tau}''$-branches;
recall~\eqref{eq:det-row} and \eqref{eq:per-row}.

The key technical difference to our earlier design in \S\ref{sect:per-e} is
that we do not use a dedicated reverse-emulation subroutine to process the 
$M_{i_1,i_2,\tau}''$-branches as in \S\ref{sect:per-e}, but rather follow 
Valiant's strategy~\cite{Valiant1979} and work on these branches essentially 
recursively, but crucially leaving rows $i_1$ and $i_2$ intact; that is, 
throughout the processing of a $M_{i_1,i_2,\tau}''$-subtree, we have 
that row $i_2$ equals $\tau$ times row $i_1$. Thus, any row operation
on distinct rows $i_1',i_2'\in [n]\setminus\{i_1,i_2\}$ of a matrix $M$ 
in such a subtree has the property that the permanent 
of the $M_{i_1',i_2',\tau'}''$-branch vanishes in $\E_{4^d}$.%
\footnote{Indeed, in such a matrix $M_{i_1',i_2',\tau'}''$ we 
have that $i_2$ is similar to $i_1$, and $i_2'$ is similar to $i_1'$, 
so the permanent vanishes in characteristic $4$ by Valiant's
observation \cite{Valiant1979}, cf.~\S\ref{sect:per-e}.}{}
Thus, each $M_{i_1,i_2,\tau}''$-subtree is in effect a non-branching 
elimination that avoids the rows $i_1$ and $i_2$. 

Another technical difference---which is crucial to gain from the ND order and 
for compatibility with the Gilbert--Tarjan analysis~\cite{GilbertT1987}---is
that we run elimination 
\begin{itemize}
\item[(a)]
in the ND order
$1,2,\ldots,n$ (omitting $i_1$ and $i_2$ from the order 
in each $M_{i_1,i_2,\tau}''$-subtree) and restricted to $i_2>i_1$ 
(respectively, $i_2'>i_1'$ with $i_1',i_2'\in [n]\setminus\{i_1,i_2\}$
for each $M_{i_1,i_2,\tau}''$-subtree) 
to obtain an intermediate matrix with odd 
entries, if any, only in 
the upper triangle;  
\item[(b)] then further eliminate the intermediate matrix in the reverse 
ND order $n,n-1,\ldots,1$ (omitting $i_1$ and $i_2$ from the order in each 
$M_{i_1,i_2,\tau}''$-subtree) and now with $i_2<i_1$ 
(respectively, $i_2'<i_1'$ with $i_1',i_2'\in [n]\setminus\{i_1,i_2\}$
for each $M_{i_1,i_2,\tau}''$-subtree) to obtain 
a leaf matrix with odd entries, if any, only on the diagonal 
(respectively, only on the diagonal as well as 
rows $i_1,i_2$ as well as columns $i_1,i_2$---a good way to visualize 
this allowed pattern of odd entries, if any, is to take a ``\#''-pattern 
and insert the diagonal); and
\item[(c)] 
computing the permanent of the diagonal matrix as the product of 
its diagonal entries (respectively, using a dedicated subroutine---described
in what follows---to compute the permanent of a ``diagonal-and-\#''-patterned 
matrix with row $i_2$ similar to row $i_1$).
\end{itemize}
Slightly less precisely, in (a) we essentially follow standard Gaussian 
elimination with diagonal pivoting to reduce to an upper-triangular matrix,
then in (b) we reduce the upper-triangular matrix to a diagonal matrix, 
at which point (c) the permanent is a product of diagonal entries---whenever
we apply a branching row operation \eqref{eq:per-row} on rows $i_1$ and $i_2$, 
we apply a similar Gaussian elimination strategy to the matrix in the 
$M_{i_1,i_2,\tau}''$-branch, but we do not touch the rows $i_1$ and $i_2$;
accordingly, the reduced matrix is ``diagonal-and-\#''-patterned rather than
diagonal, and we resort to a dedicated subroutine for computing its permanent.

Before analysing the running time of the elimination phases (a) and (b), 
as well as completing the permanent subroutine for (c), let us analyse 
the failure probability of the elimination procedure in terms of the random 
choices of the values $\beta_{uv}\in\F_{2^d}$ in (D2). Indeed, we observe that 
pivot-free elimination fails when a diagonal element $\sigma_{i,i}\in\E_{4^d}$ 
is not odd and we are applying $\sigma_{i,i}$ in a row operation 
with $i_1=i$ (respectively, $i_1'=i$) during (a). Furthermore, such a failure 
can occur only during phase (a) because phases (b) and (c) do not modify 
diagonal elements from the values they stabilise to in phase (a).
By the structure of phase (a), we observe that 
$\sigma_{i,i}=\bar\beta_{ii}+\eta$, where $\eta$ is an expression that
depends on the choices of $\beta_{i'j'}$ for ND-order-relabeled input graph 
arcs $(i',j')\in [i]\times [i]\setminus\{(i,i)\}$, but in particular $\eta$ 
is independent of $\beta_{ii}$. For any fixed $\eta\in\E_{4^d}$, we thus have 
that $\bar\beta_{ii}+\eta$ is even with probability $2^{-d}$.
By the union bound on the $n$ diagonal elements in each of the matrices
considered, of which there are at most $1+n(n-1)\leq n^2$---namely the original 
input matrix and the matrices created by the at most $n(n-1)$ branching 
row operations when reducing the input matrix---we have that the 
probability that the elimination procedure fails is at most $2^{-d}n^3$.
By our choice of $d$ in (D1), this is at most $O(n^{-1})$.

Let us now proceed to analyse the running time of phases (a) and (b). 
First, we observe that phase (a) falls under the 
Gilbert--Tarjan~\cite{GilbertT1987} analysis of Gaussian elimination
(or more precisely, odd elimination in our case, 
cf.~Lemma~\ref{lem:odd-elimination}) to triangular form in ND order. 
In particular, since Theorem~\ref{thm:gilbert-tarjan} applies to bounded 
genus graphs, we observe that the fill for each matrix considered in phase (a) 
is at most $O(\delta n\log n)$ by Theorem~\ref{thm:gilbert-tarjan}. Thus, we 
can improve the earlier upper bound on the number of branching row operations 
from $n(n-1)$ to $O(\delta n\log n)$ since each element of the fill is 
associated with at most one row operation. Accordingly, the number of 
$M_{i_1,i_2,\tau}''$-subtrees considered in phase (a) is at most 
$O(\delta n\log n)$. Processing one such $M_{i_1,i_2,\tau}''$-subtree 
in phase (a) leads to $O(n^{3/2})$ arithmetic operations in $\E_{4^d}$; 
indeed, this follows by the Gilbert--Tarjan operation-count analysis 
for bounded genus graphs in \cite[Corollary~1]{GilbertT1987} and the fact 
that each operation in the Gilbert--Tarjan analysis translates to at most 
$O(1)$ operations in our case---namely, the original operation as well as 
operations on entries of {\em columns} $i_1$ and $i_2$ that must be 
maintained under elimination since {\em rows} $i_1$ and $i_2$ are not touched. 
Since the same analysis bounds the total number of $\E_{4^d}$-arithmetic 
operations done on the branching row operations, we have that the total 
number of $\E_{4^d}$-arithmetic operations in phase (a) is 
$O(\delta n^{\frac{5}{2}}\log n)$. Due to the upper-triangular 
structure in phase (b), and the fill at most $O(\delta n\log n)$ for 
each of the at most $O(\delta n\log n)$ matrices considered in phase (b), 
we have that the total number of $\E_{4^d}$-arithmetic operations in phase (b)
is at most $O((\delta n\log n)^2)$. Thus, since $d=O(\log n)$, and $\delta$ 
is a constant for bounded genus graphs~(cf.~\cite{GilbertT1987}), 
phases (a) and (b) run in time $\tilde O(n^{2+\frac{1}{2}})$ for
bounded genus graphs.

It remains to complete phase (c) and analyse its running time. 
Let $L\in\E_{4^d}^{n\times n}$ be a matrix that is odd at the diagonal 
and may have odd entries at rows $i_1,i_2$ as well as at columns $i_1,i_2$
for distinct $i_1,i_2\in[n]$. Let us write $\sigma_{i,j}$ for the entry 
of $L$ at row $i\in [n]$, column $j\in [n]$. First, suppose that $L$ is odd 
only at the diagonal. 
Then, $\per L$ is the product of the diagonal entries---indeed, 
consider an arbitrary permutation $f\in S_n$ in~\eqref{eq:per}, and observe 
that either $f$ is the identity permutation
or $f$ moves at least two points; since only the diagonal is odd, the 
latter case translates to a monomial 
$\sigma_{1,f(1)}\sigma_{2,f(2)}\cdots \sigma_{n,f(n)}$
in~\eqref{eq:per} with at least two even terms, which vanishes in $\E_{4^d}$. 
Next, suppose that $L$ has at most $O(n^{1/2})$ odd entries in rows $i_1$
and $i_2$, and furthermore that row $i_2$ is similar to row $i_1$ in $L$.
We first show that in this case it suffices to consider permutations
$f\in S_n$ that touch only odd entries of $L$. Consider an arbitrary 
permutation $f\in S_n$. Suppose that there exists an $i\in [n]$ such that 
$\sigma_{i,f(i)}$ is even. Recall that in $\E_{4^d}$ the result 
of a multiplication with at least one even operand is even. Construct the 
permutation $f':[n]\rightarrow[n]$ as in \eqref{eq:f-prime}. Since 
row $i_2$ is similar to row $i_1$, we have that $\sigma_{i,f'(i)}$ is even,
and, furthermore, \eqref{eq:sim-monomial} holds by the reasoning in the
proof of Lemma~\ref{lem:simpair-per-det}. Thus, since the product of two
even elements vanishes in $\E_{4^d}$, we conclude that
\[
\sigma_{1,f(1)}\sigma_{2,f(2)}\cdots \sigma_{n,f(n)}+
\sigma_{1,f'(1)}\sigma_{2,f'(2)}\cdots \sigma_{n,f'(n)}=
2\sigma_{1,f(1)}\sigma_{2,f(2)}\cdots \sigma_{n,f(n)}=0\,,
\]
and hence only the permutations $f\in S_n$ that touch only odd entries
of $L$ have monomials that give a potentially nonzero contribution to $\per L$.
Thus, to compute $\per L$ it suffices to iterate over such permutations 
$f\in S_n$ and sum the contributions of their monomials 
$\sigma_{1,f(1)}\sigma_{2,f(2)}\cdots \sigma_{n,f(n)}$.
We iterate over such $f\in S_n$ by first considering all possible
images $f(i_1)=j_1$ and $f(i_2)=j_2$ such that both 
$\sigma_{i_1,j_2}$ and $\sigma_{i_2,j_2}$ are odd. By our assumption
on $L$, there are at most $O(n)$ such choices; furthermore, for
each such choice, we must have $f(j_1)\in\{i_1,i_2\}$ and 
$f(j_2)\in\{i_1,i_2\}$ or otherwise an even element is touched;
choosing $f(j_1)$ and $f(j_2)$ accordingly (unless not already chosen), 
we must have $f(i)=i$ for all elements $i\in [n]$ whose image
is not yet fixed. This leads to at most $O(n)$ permutations $f\in S_n$ 
to be iterated over. By preprocessing the products of diagonal elements
of $L$ into a perfect binary tree of subproducts (each internal node
is the product of its child nodes; the leaves are the diagonal elements,
padded with $1$-elements to get the least power of two at least $n$), we
can compute the monomial of each $f$ in the iteration in $O(\log n)$
arithmetic operations in $\E_{4^d}$. Thus, since $d=O(\log n)$, for a 
given $L$ meeting our assumptions we can compute $\per L$ in 
time $\tilde O(n)$. Taken over all the $O(\delta n\log n)$ matrices arriving
to phase (c) from phases (a) and (b), 
this translates to $\tilde O(n^2)$ total time for phase (c).

It remains to justify our assumption that each matrix $L$ with row
$i_2$ similar to row $i_1$ arriving from phases (a) and (b) to phase (c)
has the property that both row $i_1$ and row $i_2$ have at most $O(n^{1/2})$
odd entries. Since row $i_2$ by definition equals some coefficient 
times row $i_1$, it suffices to show this for row $i_1$. For phase (b), 
row $i_1$ has exactly one odd entry since we are eliminating an upper
triangular matrix to a diagonal matrix in order $n,n-1,...,1$. 
For phase (a), let us recall the recursive structure of the ND order 
reviewed after Theorem~\ref{thm:gilbert-tarjan}. Namely, for $n>n_0$,
at top level of recursion we have a partition of $[n]$ into 
sets $A_1,A_2,\ldots,A_t,C\subseteq [n]$ with $A_1<A_2<\cdots<A_t<C$ such 
that the input matrix $M$ (relabeled to ND order as per our assumption) has 
the ``diagonal-and-hook'' block structure
\begin{equation}
\label{eq:nd-diag-hook}
\begin{array}{c|ccccc}
    & A_1 & A_2 & \cdots & A_t & C\\ \hline
A_1 &  O  &     &        &     & O \\
A_2 &     &  O  &        &     & O \\
\vdots &  &     & \ddots &     & \vdots \\
A_t &     &     &        &  O  & O \\
C   &  O  &  O  & \cdots &  O  & O
\end{array}\,,
\end{equation}
where the symbol ``$O$'' indicates blocks that may contain odd entries, 
all other blocks are even. This structure is then further refined by 
recursing into each $A_j$ for $j=1,2,\ldots,t$ to obtain a tree 
of separators with $C$ at the root. Let us prove by double induction on the 
height of this tree and the size parameter $n$ of $M$ that whenever 
a row operation $(i_1,i_2,\tau)$ executed in phase (a), the row $i_1$ has at 
most $c''n^{1/2}$ odd entries for a constant $c''>0$ to be selected. 
For the base case, a tree of height one has $1\leq n\leq n_0$, so the base
case holds if $c''\geq n_0^{1/2}$. So suppose the claim holds for trees
of height $h\geq 1$, and consider a tree of height $h+1$. We may assume that
$n>n_0$; otherwise the reasoning in the base case applies. Thus, $M$ has 
the structure \eqref{eq:nd-diag-hook} with $|A_j|\leq cn$ and 
$|C|\leq c'n^{1/2}$. Recall the structure of the elimination in phase (a).
First, suppose that $i_1\in A_j$ for some $j=1,2,\ldots,t$. Applying
the induction hypothesis to the subtree of $A_j$ with height at most 
$h$ and the size parameter $|A_j|\leq cn$, we conclude that row $i_1$ has 
at most $c''|A_j|^{1/2}+|C|\leq c''(cn)^{1/2}+c'n^{1/2}\leq c''n^{1/2}$ 
odd entries if $c''\geq c'/(1-c^{1/2})$; here the term 
$c''|A_j|^{1/2}$ comes from recursive elimination inside 
the $(A_j,A_j)$-block in \eqref{eq:nd-diag-hook}, and the term
$|C|$ comes from the $C$-column in \eqref{eq:nd-diag-hook}.
Second, suppose that $i_1\in C$. At this point in elimination, 
each $(A_j,C)$-block has become even, so we have that row $i_1$ 
has at most $|C|\leq c'n^{1/2}\leq c''n^{1/2}$ odd entries if
$c''\geq c'$. We conclude that the claim holds when we take
$c''=\max\bigl(n_0^{1/2},c'/(1-c^{1/2})\bigr)$. This completes the
description and analysis of the branching elimination procedure
for the permanent over $\E_{4^d}$ in ND order; that is, 
the subroutine in (D5) of \S\ref{sect:rand-even-cycle}.

\section{Shortest two disjoint paths in undirected graphs}

\label{sect:s2dp}

This section establishes the following corollary of the present techniques
when combined with techniques of Bj\"orklund and Husfeldt~\cite{BjorklundH19}
for the shortest two disjoint paths problem.

\begin{Thm}[Shortest two disjoint paths]%
  \label{thm:s2dp}
  Given as input an undirected, unweighted, $n$-vertex graph $G$ together with terminals $s_1,t_1,s_2,t_2\in V(G)$, there is an algorithm with running time $\tilde O(n^{3+\omega})$ that with probability $1-O(n^{-1})$ determines the shortest total length of any pair of vertex-disjoint paths $P_1$ and $P_2$ in $G$
with $s_1,t_1\in V(P_1)$ and $s_2,t_2\in V(P_2)$.
\end{Thm}

We sketch the proof based on the constructions of Bj\"orklund and Husfeldt~\cite{BjorklundH19}. Without loss of generality we can assume that the undirected graph $G$ has a loop at every vertex. Let us work over a ring $R$ and associate a weight $w_{\{u,v\}}\in R$ with every edge $\{u,v\}\in E(G)$. Define the weighted \emph{symmetric} adjacency matrix $A$
such that the entry at row $u\in V(G)$ and column $v\in V(G)$ is defined by
\[
  A_{u,v} = A_{v,u} =
  \begin{cases}
    w_{\{u\}}   & \text{if $u=v$};\\
    w_{\{u,v\}} & \text{if $u\neq v$ and $\{u,v\}\in E(G)$};\\
    0 & \text{otherwise}.
  \end{cases}
\]
For a subset $U\subseteq V(G)$ of vertices, let us write $A_U$ for the matrix 
obtained from $A$ by deleting the rows and columns corresponding to $U$.
Define the \emph{disjoint paths enumerator} of $A$ as 
\[
  \operatorname{dp} A = 
  \sum_{P_1,P_2}
  \left(\prod_{uv\in P_1\cup P_2} A_{u,v}\right)
  \per A_{V(P_1)\cup V(P_2)}\,,
\]
where the sum is over all vertex disjoint paths $P_1$ and $P_2$ in $G$
with $s_1,t_1\in V(P_1)$ and $s_2,t_2\in V(P_2)$.

To obtain an algebraic fingerprint, 
follow the analog of~\eqref{eq:adj-mat-y} and 
extend $A$ to a matrix $A_y\in R[y]^{n\times n}$ in the indeterminate $y$ 
by multiplying the {\em non-loops} with $y$.
To be concrete,
\[
  (A_y)_{u,v} =
  (A_y)_{v,u} =
  \begin{cases} 
    w_{\{u\}} & \text{if $u=v$};\\
    yw_{\{u,v\}} & \text{if $u\neq v$ and $\{u,v\}\in E(G)$};\\
    0 & \text{otherwise}.
  \end{cases}
\]
The reasoning behind Theorem~1.1 of \cite{BjorklundH19} then establishes 
that $G$ contains a unique pair of disjoint paths of total 
length $k$ if and only if $[y^k]\operatorname{dp} A_y$ is the lowest-order 
term of $\operatorname{dp} A_y$, viewed as a polynomial in $y$, that 
is not identically zero.

Defining $A[vw, v'w']$ as in \cite[Equation~(1.2)]{BjorklundH19}, the central 
characterisation in \cite[Lemma~2.1]{BjorklundH19} 
with $f = 2\operatorname{dp}$  becomes
\begin{equation}\label{eq:dp}
  2\operatorname{dp} A  =  \per A[t_1s_1, t_2s_2] + \per A[t_1s_1, s_2t_2] - \per A[s_1s_2, t_1t_2] 
\,.
\end{equation}
This expression, like \eqref{eq:two-pcc}, is a multivariate polynomial 
identity of the form $2p=q$, so the same approach as in the present paper 
works. In particular, the disjoint paths enumerator $\operatorname{dp}$ can 
be evaluated much like the parity cycle cover enumerator $\pcc_{n-1}$, 
as described in \S\ref{sect:pcc-ff}.

\medskip

The improvements to the running time in \S\ref{sect:nested} apply here 
as well. In particular, the running time for bounded genus instances 
becomes $\tilde O(n^{3+\frac{1}{2}})$ as in Theorem~\ref{thm:boundedgenus}.
In particular, 
we observe that the two directed edges added to the three graphs underlying 
the matrices in \eqref{eq:dp} increase the genus by at most $2$.


\bibliographystyle{abbrv}
\bibliography{paper}

\begin{thebibliography}{10}

\bibitem{AlonY2013}
N.~Alon and R.~Yuster.
\newblock Matrix sparsification and nested dissection over arbitrary fields.
\newblock {\em J. {ACM}}, 60(4):25:1--25:18, 2013.

\bibitem{ArkinPY91}
E.~M. Arkin, C.~H. Papadimitriou, and M.~Yannakakis.
\newblock Modularity of cycles and paths in graphs.
\newblock {\em J. ACM}, 38(2):255–274, Apr. 1991.

\bibitem{Berkowitz1984}
S.~J. Berkowitz.
\newblock On computing the determinant in small parallel time using a small
  number of processors.
\newblock {\em Inform. Process. Lett.}, 18(3):147--150, 1984.

\bibitem{BjorklundH19}
A.~Bj{\"{o}}rklund and T.~Husfeldt.
\newblock Shortest two disjoint paths in polynomial time.
\newblock {\em {SIAM} J. Comput.}, 48(6):1698--1710, 2019.

\bibitem{BunchH1974}
J.~R. Bunch and J.~E. Hopcroft.
\newblock Triangular factorization and inversion by fast matrix multiplication.
\newblock {\em Math. Comp.}, 28:231--236, 1974.

\bibitem{ChungGK94}
F.~R.~K. Chung, W.~Goddard, and D.~J. Kleitman.
\newblock Even cycles in directed graphs.
\newblock {\em {SIAM} J. Discret. Math.}, 7(3):474--483, 1994.

\bibitem{DeMilloL1978}
R.~A. DeMillo and R.~J. Lipton.
\newblock A probabilistic remark on algebraic program testing.
\newblock {\em Inform. Process. Lett.}, 7(4):193--195, 1978.

\bibitem{George1973}
A.~George.
\newblock Nested dissection of a regular finite element mesh.
\newblock {\em SIAM J. Numer. Anal.}, 10:345--363, 1973.

\bibitem{GilbertT1987}
J.~R. Gilbert and R.~E. Tarjan.
\newblock The analysis of a nested dissection algorithm.
\newblock {\em Numer. Math.}, 50(4):377--404, 1987.

\bibitem{GrotschelP81}
M.~Gr{\"{o}}tschel and W.~R. Pulleyblank.
\newblock Weakly bipartite graphs and the {M}ax-cut problem.
\newblock {\em Oper. Res. Lett.}, 1(1):23--27, 1981.

\bibitem{GutinSY1998}
G.~Z. Gutin, B.~Sudakov, and A.~Yeo.
\newblock Note on alternating directed cycles.
\newblock {\em Discrete Math.}, 191(1-3):101--107, 1998.

\bibitem{HemaspaandraST04}
E.~Hemaspaandra, H.~Spakowski, and M.~Thakur.
\newblock Complexity of cycle length modularity problems in graphs.
\newblock In M.~Farach{-}Colton, editor, {\em {LATIN} 2004: Theoretical
  Informatics, 6th Latin American Symposium, Buenos Aires, Argentina, April
  5-8, 2004, Proceedings}, volume 2976 of {\em Lecture Notes in Computer
  Science}, pages 509--518. Springer, 2004.

\bibitem{ItaiR78}
A.~Itai and M.~Rodeh.
\newblock Finding a minimum circuit in a graph.
\newblock {\em {SIAM} J. Comput.}, 7(4):413--423, 1978.

\bibitem{Kaltofen1992}
E.~Kaltofen.
\newblock On computing determinants of matrices without divisions.
\newblock In P.~S. Wang, editor, {\em Proceedings of the 1992 International
  Symposium on Symbolic and Algebraic Computation, {ISSAC} '92, Berkeley, CA,
  USA, July 27-29, 1992}, pages 342--349. {ACM}, 1992.

\bibitem{KoutisW2016}
I.~Koutis and R.~Williams.
\newblock Algebraic fingerprints for faster algorithms.
\newblock {\em Comm. {ACM}}, 59(1):98--105, 2016.

\bibitem{LabahnNZ2017}
G.~Labahn, V.~Neiger, and W.~Zhou.
\newblock Fast, deterministic computation of the {H}ermite normal form and
  determinant of a polynomial matrix.
\newblock {\em J. Complexity}, 42:44--71, 2017.

\bibitem{LidlN1997}
R.~Lidl and H.~Niederreiter.
\newblock {\em Finite Fields}, volume~20 of {\em Encyclopedia of Mathematics
  and its Applications}.
\newblock Cambridge University Press, Cambridge, second edition, 1997.

\bibitem{LiptonT1979}
R.~J. Lipton and R.~E. Tarjan.
\newblock A separator theorem for planar graphs.
\newblock {\em SIAM J. Appl. Math.}, 36(2):177--189, 1979.

\bibitem{Little1975}
C.~Little.
\newblock A characterization of convertible (0, 1)-matrices.
\newblock {\em J. Combin. Theory Ser. B}, 18(3):187--208, 1975.

\bibitem{Lubiw88}
A.~Lubiw.
\newblock A note on odd/even cycles.
\newblock {\em Discrete Appl. Math.}, 22(1):87--92, 1988.

\bibitem{McCuaig00}
W.~McCuaig.
\newblock Even dicycles.
\newblock {\em J. Graph Theory}, 35(1):46--68, 2000.

\bibitem{McCuaig04}
W.~McCuaig.
\newblock P{\'{o}}lya's permanent problem.
\newblock {\em Electron. J. Combin.}, 11(1), 2004.

\bibitem{McCuaigRST97}
W.~McCuaig, N.~Robertson, P.~D. Seymour, and R.~Thomas.
\newblock Permanents, {P}faffian orientations, and even directed circuits
  (extended abstract).
\newblock In F.~T. Leighton and P.~W. Shor, editors, {\em Proceedings of the
  Twenty-Ninth Annual {ACM} Symposium on the Theory of Computing, El Paso,
  Texas, USA, May 4-6, 1997}, pages 402--405. {ACM}, 1997.

\bibitem{Monien83}
B.~Monien.
\newblock The complexity of determining a shortest cycle of even length.
\newblock {\em Computing}, 31(4):355--369, 1983.

\bibitem{RST99}
N.~Robertson, P.~D. Seymour, and R.~Thomas.
\newblock Permanents, {P}faffian orientations, and even directed circuits.
\newblock {\em Ann. of Math. (2)}, 150:929--975, 1999.

\bibitem{Schwartz1980}
J.~T. Schwartz.
\newblock Fast probabilistic algorithms for verification of polynomial
  identities.
\newblock {\em J. {ACM}}, 27(4):701--717, 1980.

\bibitem{SeymourT1987}
P.~D. Seymour and C.~Thomassen.
\newblock Characterization of even directed graphs.
\newblock {\em J. Combin. Theory Ser. B}, 42(1):36--45, 1987.

\bibitem{Thomassen85}
C.~Thomassen.
\newblock Even cycles in directed graphs.
\newblock {\em European J. Combin.}, 6(1):85--89, 1985.

\bibitem{Thomassen88}
C.~Thomassen.
\newblock On the presence of disjoint subgraphs of a specified type.
\newblock {\em J. Graph Theory}, 12(1):101--111, 1988.

\bibitem{Thomassen92}
C.~Thomassen.
\newblock The even cycle problem for directed graphs.
\newblock {\em J. Amer. Math. Soc.}, 5(2):217--229, 1992.

\bibitem{Thomassen93}
C.~Thomassen.
\newblock The even cycle problem for planar digraphs.
\newblock {\em J. Algorithms}, 15(1):61--75, 1993.

\bibitem{Valiant1979}
L.~Valiant.
\newblock The complexity of computing the permanent.
\newblock {\em Theoret. Comput. Sci.}, 8(2):189--201, 1979.

\bibitem{VaziraniY1989}
V.~V. Vazirani and M.~Yannakakis.
\newblock Pfaffian orientations, {$0$}-{$1$} permanents, and even cycles in
  directed graphs.
\newblock {\em Discrete Appl. Math.}, 25(1-2):179--190, 1989.

\bibitem{vonzurGathenG2013}
J.~von~zur Gathen and J.~Gerhard.
\newblock {\em Modern Computer Algebra}.
\newblock Cambridge University Press, third edition, 2013.

\bibitem{Younger73}
D.~H. Younger.
\newblock Graphs with interlinked directed circuits.
\newblock In {\em Proc. Midwest Symp. on Circuit Theory}, volume~2, pages
  XVI2.1--XVI2.7, 1973.

\bibitem{Yuster2008}
R.~Yuster.
\newblock Matrix sparsification for rank and determinant computations via
  nested dissection.
\newblock In {\em 49th Annual {IEEE} Symposium on Foundations of Computer
  Science, {FOCS} 2008, October 25-28, 2008, Philadelphia, PA, {USA}}, pages
  137--145. {IEEE} Computer Society, 2008.

\bibitem{YusterZ97}
R.~Yuster and U.~Zwick.
\newblock Finding even cycles even faster.
\newblock {\em {SIAM} J. Discrete Math.}, 10(2):209--222, 1997.

\bibitem{Zippel1979}
R.~Zippel.
\newblock Probabilistic algorithms for sparse polynomials.
\newblock In E.~W. Ng, editor, {\em Symbolic and Algebraic Computation,
  {EUROSAM} '79, An International Symposium on Symbolic and Algebraic
  Computation, Marseille, France, June 1979, Proceedings}, volume~72 of {\em
  Lecture Notes in Computer Science}, pages 216--226. Springer, 1979.

\end{thebibliography}






\end{document}